\documentclass[journal]{IEEEtran}
\ifCLASSINFOpdf

\else
\fi

\usepackage[svganames]{xcolor}
\usepackage{graphicx}
\usepackage{amssymb}
\usepackage{flushend,blindtext}
\usepackage{amsmath, color}
\usepackage{amsthm}
\usepackage{subfig}
\usepackage{url}

\usepackage{hyperref}%
\hypersetup{colorlinks=true, linkcolor=blue, breaklinks=true, urlcolor=blue}

\usepackage{lipsum}
\usepackage{mathtools}
\usepackage{cuted}
\usepackage{float}
\usepackage{doi,xspace}
\usepackage[font=small]{caption}
\usepackage{graphicx}
\usepackage{lipsum}
\usepackage{multicol} %% To use two columns equally filled
\usepackage{booktabs}
\usepackage{cuted}
\usepackage{enumerate}

\usepackage{xcolor}
\definecolor{Emerald}{RGB}{40, 90, 40}

\usepackage[linesnumbered,algoruled,boxed,lined]{algorithm2e}
\makeatletter 
\g@addto@macro{\@algocf@init}{\SetKwInOut{Parameter}{Parameters}} 
\makeatother

\newcommand{\Edges}{\mathcal{E}}

\newcommand{\until}[1]{\{1,\dots, #1\}}

\newcommand\ddfrac[2]{\frac{\displaystyle #1}{\displaystyle #2}}

\newcommand\oprocendsymbol{\hbox{$\square$}}
\newcommand\oprocend{\relax\ifmmode\else\unskip\hfill\fi\oprocendsymbol}

\usepackage[prependcaption,colorinlistoftodos]{todonotes}

\renewcommand{\theenumi}{(\roman{enumi})}

\DeclareSymbolFont{bbold}{U}{bbold}{m}{n}
\DeclareSymbolFontAlphabet{\mathbbold}{bbold}
\newcommand{\vect}[1]{\mathbbold{#1}}
\newcommand{\vectorones}[1][]{\vect{1}_{#1}}

\DeclareSymbolFont{bbold}{U}{bbold}{m}{n}
\DeclareSymbolFontAlphabet{\mathbbold}{bbold}
\newtheorem{theorem}{Theorem}[section]

\newtheorem{corollary}[theorem]{Corollary}
\newtheorem{definition}[theorem]{Definition}
\newtheorem{lemma}[theorem]{Lemma}

\newtheorem{remark}[theorem]{Remark} 

\newtheorem{example}[theorem]{Example}

\graphicspath{{./Figures/}}

\begin{document}

\title{Network Formation for Multigroup Coordination: \newline
Stable vs.\ Efficient Connections}

\title{Stable and Efficient Structures \\ in Multigroup Network Formation}

\author{Shadi Mohagheghi,
        Jingying Ma,
        and~Francesco Bullo% <-this % stops a space
\thanks{Shadi Mohagheghi and Francesco Bullo are with the Department of Electrical and Computer Engineering and the Center for Control, Dynamical Systems and Computation, University of California at Santa Barbara, Santa Barbara, 93106-9560 California, 
USA, \texttt{\{shadi,bullo\}@ucsb.edu} }% <-this % stops a space
\thanks{Jingying Ma is with the School of Mathematics and Statistics, Ningxia University, Yinchuan 750021, P.~R.~China, \texttt{majy1980@126.com} }% <-this % stops a space
\thanks{This material is based upon work supported by, or in part by, the U.S.~Army Research Laboratory and the U.S.~Army Research Office under grant numbers  W911NF-15-1-0577.}
}

\maketitle

\begin{abstract}
  In this work we present a strategic network formation model predicting
  the emergence of multigroup structures. Individuals decide to form or
  remove links based on the benefits and costs those connections carry; we
  focus on bilateral consent for link formation. An exogenous system
  specifies the frequency of coordination issues arising among the
  groups. We are interested in structures that arise to resolve
  coordination issues and, specifically, structures in which groups are
  linked through bridging, redundant, and co-membership interconnections.
  We characterize the conditions under which certain structures are stable
  and study their efficiency as well as the convergence of formation
  dynamics.
\end{abstract}

\begin{IEEEkeywords}
  strategic network formation, game theory, multigroup connectivity models
\end{IEEEkeywords}

\IEEEpeerreviewmaketitle

\section{Introduction}

\subsection{Motivation and problem description}

\IEEEPARstart{T}o study the coordination and control features of a group
task, the multiple groups’ performances must be fitted together. 
 An enduring postulate in organization science
is that coordination and control cannot be achieved strictly by the
authority structure, but must also entail informal communication and
influence networks that link the members of different task-oriented groups;
we focus on formation of such network structures. As the size of a
connected social network increases, multigroup formations that are
distinguishable clusters of individuals become a characteristic and
important feature of network topology. The connectivity of multigroup
networks may be based on edge bundles connecting multiple individuals in
two disjoint groups, bridges connecting two individuals in two disjoint
groups, or co-memberships. A large-scale network may include instances of
all of these connectivity modalities.  We set up populations of multiple
groups and propose a dynamic model for formation of these intergroup
connectivity structures.

Our economic dynamical model explains and predicts whether a network
evolves into different coordination and control structures. Medium and
large scale organizations adopt these multigroup structures to tackle
complex nested tasks. Among the multitude of possible coordination and
control structures, we study formation of multigroup connectivity
structures shown in Fig.~\ref{fig:schematic}, which are familiar constructs
in the field of social network science.
\begin{figure}[h]
	\begin{center} 
		\subfloat[Co-memberships]{\includegraphics[height=.89in]{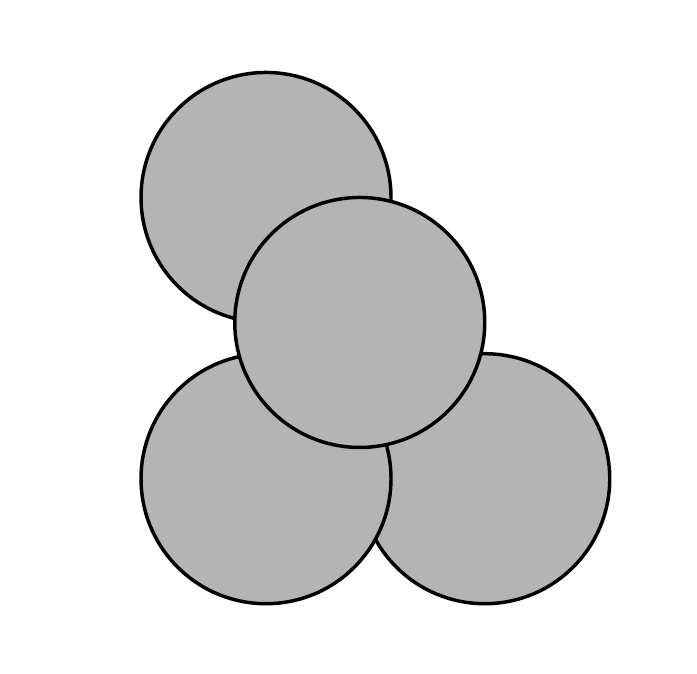}\label{fig:control-structure-3}}\qquad
		\subfloat[Edge Bundles]{\includegraphics[height=.8in]{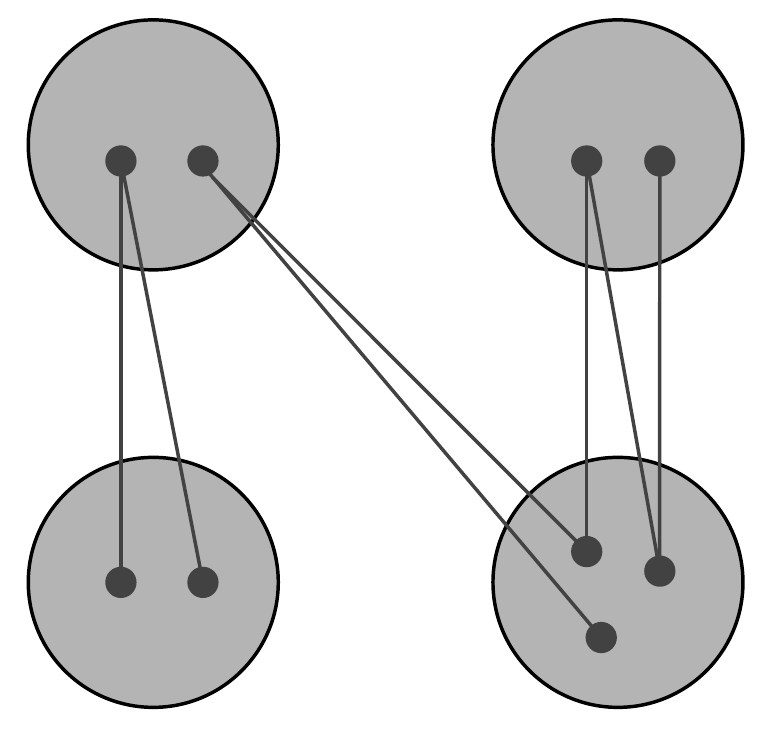}\label{fig:control-structure-2}}\qquad
		\subfloat[Bridges]{\includegraphics[height=.8in]{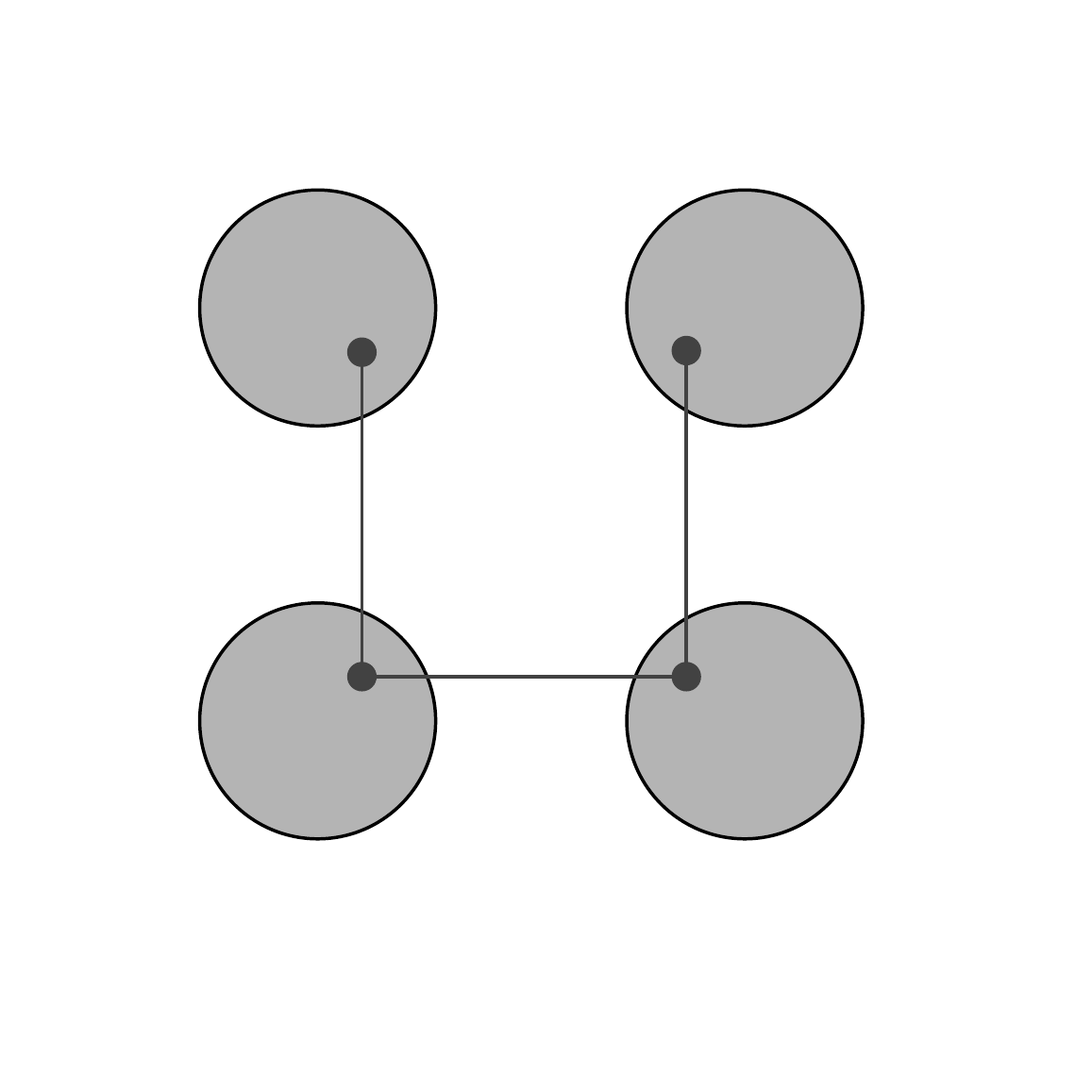}\label{fig:control-structure-1}}\\
		\caption{\small Schematic illustration of the three possible control
and coordination structures}\label{fig:schematic}
	\end{center}
\end{figure}
For this purpose we apply a game-theoretic framework in which strategic agents take actions based on the rate or importance of coordination problems. In other words, a value is assigned to the coordination problem between any two distinct groups, so that all control and coordination problems among groups are described by a square non-negative matrix, as illustrated in Fig.~\ref{fig:matrix-F}.
\begin{figure}
	\centering
	\includegraphics[width=0.3\linewidth , height=0.3\linewidth]{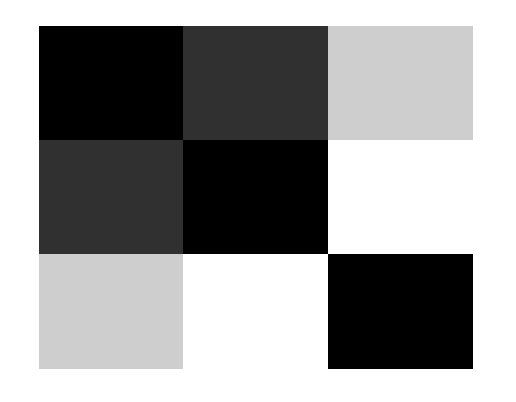}\\
\caption{$F$ = frequency/importance of intergroup coordination problem}\label{fig:matrix-F} 
\end{figure} 
In our setting, agents are myopic, self-interested, and have thorough knowledge of graph topology and the utility they acquire from any other agent.

\subsection{Related literature}

Bridge, edge bundle, and co-membership connectivity models have been studied extensively in~\cite{SM-PA-FB-NEF:17c}, where implications of these structures are investigated and generative models are proposed for each. These prototypical structures can mitigate coordination and control loss in an organization. Coordination and control importance of bridge connected structure, in which
communication between subgroups are based on single contact edges, is the emphasis of the~\cite{MSG:73, MT-DK:10}, and~\cite{WS-TE:08} models.  %\smmargin{Chapter 8} 
Coordination and control importance of the redundant ties structure, in which multiple redundant contact edges connect pairs of groups, is the emphasis
of~\cite{NEF:98}, Chapter~8, ~\cite{NEF:83}, and~\cite{HCW-SAB-RLB:76}. Co-membership intersection
structures, in which subgroups have common members, is the emphasis of the linking-pin
model by Likert~\cite{RL:67}, as well as~\cite{BC-JAH:04} and~\cite{SPB-DSH:11}.~\cite{XZ-CW-YS-LP-HZ:17} and \cite{JY-JL:12} propose a community detection algorithm for overlapping networks.

Jackson and Wolinsky introduced a strategic network formation model in their seminal paper~\cite{MOJ-AW:96}. They studied pairwise stability, where bilateral agreement is required for link formation. Homogeneity and common knowledge of current network to all players are two assumptions in this model. Jackson and Watts studied strategic network formation in a dynamic framework in~\cite{MOJ-AW:02}. The network formation model we present in this work is closely related to~\cite{MOJ-AW:96} and~\cite{MOJ-AW:02}. Jackson and Rogers examined an economic model of network formation in \cite{MOJ-BWR:05} where agents benefit from indirect relationships. They showed that small-world features necessarily emerge for a wide set of parameters.

In \cite{VB-SG:00}, Bala and Goyal proposed a dynamic model to study Nash and strict Nash stability. In their model, starting from any initial network, each player with some positive probability plays a best response (or randomizes across them when there is more than one); otherwise the player exhibits inertia. A Markov chain on the state space of all networks is defined whose absorbing states are strict Nash networks. The authors proved that starting from any network, the dynamic process converges to a strict Nash network (i.e., the empty network or a center-sponsored star) with probability 1.

In \cite{NO-FV:13}, Olaizola and Valenciano extended the model in~\cite{VB-SG:00} and studied network formation under linking constraints. An exogenous link-constraining system specifies the admissible links. Players in the same component of the link-constraining network have common knowledge of that component. This model collapses to the unrestricted setting in~\cite{VB-SG:00} (when the underling constraining network is complete graph). The set of Nash networks is a subset of  Bala and Goyal's unrestricted Nash network sets.

In the network formation game by Chasparis and Shamma in~\cite{GCC-JSS:13} and~\cite{GCC-JSS:08}, agents form and sever unidirectional links with other nodes, and stable networks are characterized through the notion of Nash equilibrium. Pagan and D{\"o}rfler~\cite{NP-FD:19} studied network formation on directed weighted graphs and considered two notions of stability: Nash equilibrium to model purely selfish actors, and pairwise-Nash stability which combines the selfish attitude with the possibility of coordination among agents. McBride dropped the common knowledge assumption and studied the effects of limited perception (each player perceives the current network only up to a certain distance) in \cite{MMB:06}. Song and van der Schaar \cite{YS-MVDS:15} studied a dynamic network formation model with incomplete information.

Community networks and their growth into potential socially robust
structures is studied in~\cite{LM:19}. Bringmann et al. analyzed the
evolution of large networks to predict link creation among the nodes
in~\cite{BB-MB-FB-AG:10}. \cite{YJ-YW-XJ-ZZ-XC:17} studied link
inference problem in heterogeneous information networks by proposing a
knapsack-constrained inference method.

\subsection{Statement of contribution} 
We consider a strategic network formation game described by a cost of maintaining links, a benefit of having connections, and an importance of coordination problems among pre-specified groups. Our setup is a heterogeneous generalization of the famous connection model. For this game, we study the resulting multi-group  structures that are pairwise stable and socially efficient.

For this game, we also introduce a formation dynamics  whereby 
link formations require mutual consent and link removals can be initiated unilaterally. We study the conditions that give rise to formation of multigroup structures, as well as conditions which cause the multigroup structures be stable and/or efficient.
Our contributions are as follows:

We introduce certain threshold functions and provide bounds based on these functions to study pairwise stable and efficient structures. We also investigate the convergence of Formation Dynamics. For our analysis, we first focus on the structure of each group and formation of intra-connections. We particularly study the conditions which result in each group being a clique, and present results on pairwise stability, efficiency, and convergence of these cliques.

We then focus on the interconnections among those cliques. We present results on the pairwise stability and convergence to disjoint union of cliques for multigroup structures of arbitrary sizes. The rest of the analysis for density of interconnections is divided into two sections: two-group connectivity structures and multigroup connectivity structures. 

For the two group structures, we provide a complete characterization of full ranges of parameters for stability and  efficiency. We present results on the pairwise stability and efficiency of minimally connected, redundantly connected, and maximally connected structures. We identify the ranges of parameter in which the efficient and the pairwise stable structure overlap and those in which they have a conflict.

We then investigate the multigroup structures. We study the pairwise stability of minimally connected cliques along arbitrary interconnection structures. We show that for the special case of the interconnection being a star graph, it is possible to identify the boundaries of parameters for stability of all interconnections being minimally connected. We also present results on formation of redundancies and for efficiency.

\subsection{Preliminaries}

Each undirected graph is identified with the pair $\mathcal{(V, E)}$. The set of graph nodes $\mathcal{V} \neq \emptyset$ represents individuals or groups of individuals in a social network. $|\mathcal{V}|=n$ is the size of the network. The pair $(i,j)$ is called an edge and it indicates the interaction between the two individuals $i$ and $j$. The set of graph edges $\Edges$ represents the social interactions or ties among all individuals. Throughout this paper, since the individuals are unchangeable, we refer to the network $\mathcal{(V, E)}$ simply as $\Edges$. 

The density of a graph is given by the ratio of the number of its observed to possible edges, $ \ddfrac{2|\Edges|}{n(n-1)}$. In a complete graph every pair of distinct nodes is connected by an edge. We denote the complete graph of size $n$ by $K_n$. A clique is a subset of vertices of a graph in which every two distinct vertices are adjacent. We say two graphs are adjacent if they differ in precisely one edge. A path of length $k$ is a sequence of nodes $i_{1}i_{2}\dots i_{k}$ such that $\{(i_{s},i_{s+1}) \} \in \Edges$.  A walk of minimum length between two nodes is the shortest path. $d_{ij}(\Edges)$ denotes the distance between nodes $i$ and $j$, which is defined as the length of the shortest path beginning at $i$ and ending at $j$. 

\section{Multigroup Network Formation Model}
\label{sec:model}
Consider a society of $n$ individuals $\mathcal{V}$, divided into $m$ groups. The set of $m$ groups is denoted by $\until{m}, m \leq n$.
$P= \{ P_1, \dots, P_m\}$ represents the partitioning of individuals into
the groups and is a set partition of size $n$, i.e, $\mathcal{V}=
\bigcup \limits_{\gamma=1}^{m} P_{\gamma}$, and $\bigcap \limits_{\gamma=1}^{m} P_{\gamma} =
\emptyset$. We use the shorthand notation $s_{\gamma}= |P_{\gamma}|$
denoting the size of group $\gamma$. Throughout this paper, we assume that
$s_{\gamma} \geq 3$ for all $\gamma \in \until{m}$.

\textit{Group coordination importance matrix (data)}: is given as $F \in
\mathbb{R}^{m \times m}$, where $ 0 \leq F_{\alpha \beta} \leq 1$ for
$\alpha, \beta \in \until{m}$ represents importance/frequency of coordination
problem between groups $\alpha$ and $\beta$. We assume $F$ is a symmetric
matrix with diagonal entries equal to $1$.

\textit{Individual coordination importance matrix}: $\hat{F}  \in \mathbb{R}^{n \times n}$, is obtained from $F$ and the partition $P$, i.e., $\hat{F}= f(F, P)$. We construct $\hat{F}$ as follows:
  
  \begin{equation*}  
    \hat{F}_{ij} =\begin{cases}
    F_{\alpha \beta}, & i \in P_{\alpha},  j \in P_{\beta},  i \neq j \\
    0, & i=j. 
    \end{cases}
  \end{equation*}
  
  For the setting where groups are all of equal size $s$, one can write
  \begin{equation*}
     \hat{F}=F \otimes  \vectorones[s] \vectorones[s]^T -I_n 
  \end{equation*}

At edge set $\Edges$, the payoff function for individual $i \in
\mathcal{V}$ is 
\begin{equation}\label{payoff-function}
  U_i (\Edges) = \sum_{k=1}^n \hat{F}_{ik} \delta^{d_{ik}
    (\Edges)} - \sum\nolimits_{k \in N_i(\Edges) } c,
\end{equation}
where $d_{ik}(\Edges)$ is the number of steps from individual $i$
to $k$, $\delta < 1$ is the one-hop benefit, and $c$ is the cost of each
link. The value of network $\Edges$ is defined as the sum of all individuals' payoffs, i.e., $v(\Edges)=\sum_{i=1}^{n}U_{i}(\Edges)$, and it indicates the social welfare. For a given society $\mathcal{V}$ and value function $v$, $\Edges^*$ is an \emph{efficient structure} if its social welfare(value) is maximized over all possible edge sets on $\mathcal{V}$, i.e., $\Edges^*=  \arg \max\limits_{\Edges} v(\Edges)$. Given the pair $(i, j)$ in network $\Edges$, we say that individual $i$ \emph{benefits from edge} $\{(i,j)\}$ if
$
  U_i \big (\Edges\cup \{( i,j) \} \big ) > U_i \big(\Edges \setminus \{( i,j) \}
  \big).
$

\textit{Formation Dynamics}: Time periods are represented with countable
infinite set $\mathbb{N}= \{ 1, 2, \dots, t, \dots \}$.  In each period, a
pair $(i,j)$ is uniformly randomly selected and is added to, or removed
from, the network $\Edges$ according to the following rules:
\begin{itemize}

\item if $\{(i,j)\} \notin \Edges$, then it is added when its
  addition is marginally beneficial to the pair of individuals (i.e.,
  either both individuals benefit or one individual is indifferent and the
  other benefits); the edge $(i,j)$ is not added when its addition causes
  a drop in the payoff of either or both individuals or both individuals
  are indifferent towards it; and
\item if $\{(i,j)\} \in \Edges$, then $(i,j)$ is removed when its
  removal benefits at least one of the two individuals; no action is taken
  when both sides are either indifferent or benefit from the existence of
  the edge.
\end{itemize}

\begin{definition}(Pairwise Stability)\label{def:pairwise_stable}
A network $\Edges$ is \emph{pairwise stable} if,
\begin{equation*}
\begin{aligned}
& \text{for all } \{(i,j)\}\in\Edges,\\
&  \quad U_i(\Edges) \geq U_i(\Edges \setminus \{(i,j)\} ) 
\text{ and } U_j(\Edges) \geq U_j(\Edges \setminus \{( i,j) \}); \\
 & \text{and } \text{for all } \{(i,j)\} \notin \Edges, \\
 &\quad \text{if } U_i(\Edges)<U_i(\Edges\cup\{(i,j)\}), \text{ then }
U_j(\Edges)>U_j(\Edges\cup\{(i,j)\}).
\end{aligned}
\end{equation*}
\end{definition}
 \begin{remark}\label{remark:pairwise-stable}
According to Definition~\ref{def:pairwise_stable}, if the edge $(i,j)$ belongs to the pairwise stable network, removing it results in a loss for $i$ or $j$; and if the edge $(i,j)$ does not belong to the pairwise stable network, adding it makes no difference or causes loss for $i$ or $j$. 
 \end{remark}

\begin{definition}
%\citep{MOJ-AW:02}
$\Edges'$ defeats $\Edges$ if either $\Edges'=\Edges \setminus
\{(i,j)\}$ and $U_i(\Edges')>U_i(\Edges)$, or $\Edges'=\Edges \cup \{(i,j)\}$
and $U_i(\Edges') \geq U_i(\Edges)$ and $U_j(\Edges') \geq U_j(\Edges)$ with at least one inequality holding strictly. 
\end{definition}\label{def:improving-path}

\begin{lemma}\label{lem:stable-by-dynamics}
A network is pairwise stable if and only if it does not change under Formation Dynamics.
\end{lemma}
\begin{proof}
To prove necessity, we refer to Remark~\ref{remark:pairwise-stable}. According to the definition, if a network is pairwise stable, no network can defeat it, i.e., no links can be added to or severed from it. To show sufficiency, note that a network not being changed by Formation Dynamics, implies that: 
\begin{enumerate}
\item adding a link makes no difference or causes loss for at least one individual;
\item removing a link results in loss for at least one individual. 
\end{enumerate}
Therefore, the network is pairwise stable. 
\end{proof}

According to Lemma~\ref{lem:stable-by-dynamics}, if there exists some time $t^*$
such that from $t^*$ on, no additional links are added to or severed from a network by Formation Dynamics, then the network has reached the pairwise stable structure.

We define the following terms that we will frequently use throughout this paper indicating the density of the interconnections among the groups. 
\begin{definition}
We say that a society of individuals consists of {\it the disjoint union of groups} if there exists no interconnection among any pairs of groups.  
For a connected pair, we say it is
\begin{enumerate}
\item {\it minimally connected} if there exists exactly one interconnection among the pair;
\item {\it redundantly connected} if there exist at least two interconnections among the pair;
\item {\it maximally connected} if all of the possible interconnections among the pair of groups exist.
\end{enumerate}
\end{definition}
Fig. \ref{fig:density_interC} represents a schematic illustration of the terms discussed above.
\begin{remark}
A minimally connected pair corresponds to the bridge connection (Fig.~\ref{fig:control-structure-1}), redundantly connected to the ridge connection (Fig.~\ref{fig:control-structure-2}), and maximally connected to a full co-membership connection (Fig.~\ref{fig:control-structure-3}.)
\end{remark}
\begin{figure}[h]
	\begin{center} 	
	\includegraphics[width=0.99\linewidth]{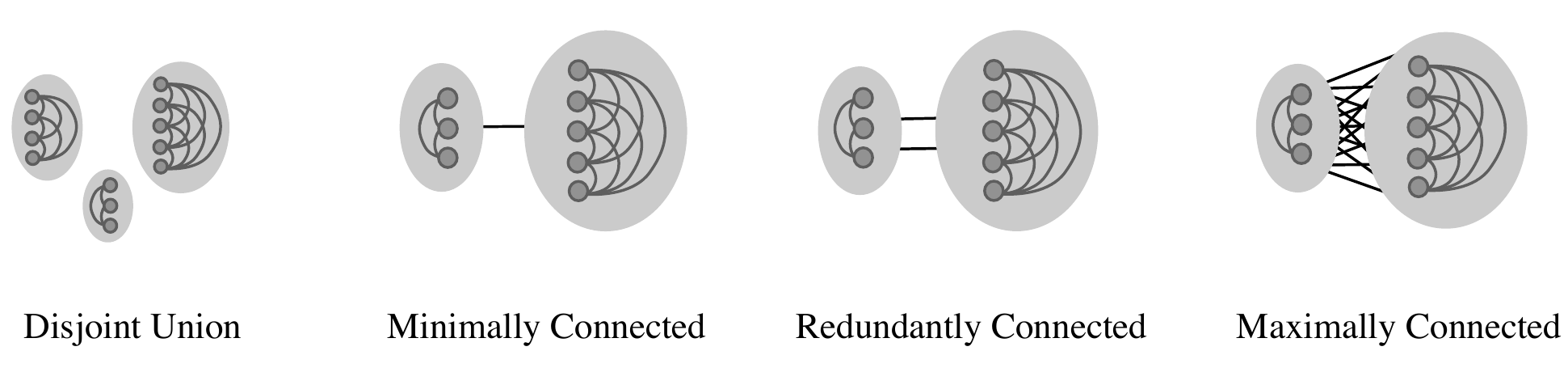}\label{fig:}
		\caption{\small Schematic illustration of interconnection densities}\label{fig:density_interC}
	\end{center}
\end{figure}

We next define the Price of Anarchy (PoA) as a measure of how the efficiency of a system degrades due to the selfish behavior of its individuals. It is calculated as follows:
\[
PoA=\dfrac{\max_{\Edges}v(\Edges)}{\min_{p.w. stable\Edges}v(\Edges)}.
\]

Throughout this paper we use the following threshold functions
$y_1(s, \delta)$, $y_2(s, \delta)$, and $y_3(\delta)$ defined by
\begin{equation*}
\begin{aligned}
  & y_1(s, \delta)= {\delta + \big(s-1\big) \delta^2},\\
  & y_2(s, \delta) =  {\delta - \delta^2 + \big(s-1\big) \delta^2 - \big(s-1\big) \delta^3} = { \big (1-\delta \big )y_1(s)},\qquad \\
 & y_3(\delta)={\delta - \delta^2}.
\end{aligned}
\end{equation*}
In what follows we will often suppress the argument $\delta$ in the interest of simplicity.

Under the conditions $0<\delta<1$ and $s\geq 3$, we claim that,
\begin{equation*}
  0<y_3<y_2(s)<y_1(s).
\end{equation*}
The proof is as follows: it is easy to see that $y_2(s) < y_1(s)$. To
verify $y_3 < y_2(s) $, we rewrite $y_2(s)$ as $ \big (\delta - \delta^2
\big) \big(1+ \delta (s-1) \big) = y_3 \big(1+ \delta (s-1) \big) > y_3$.
Plots of these three threshold functions for $0<\delta<1$, where $s=3$ are depicted in
Fig.~\ref{fig:thresholds}.
\begin{figure}[ht]
	\centering
	\includegraphics[width=0.6\linewidth]{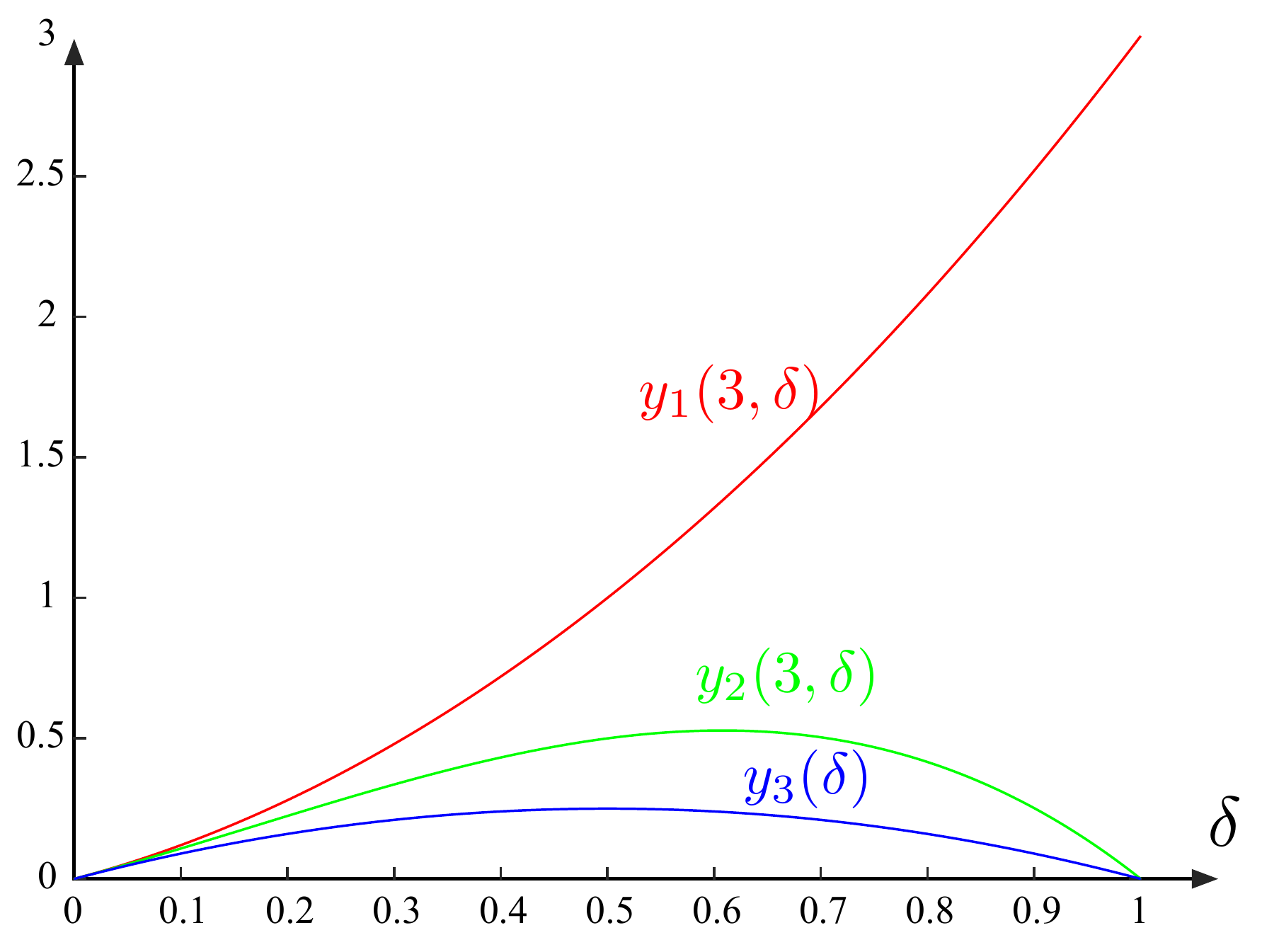}
        \caption{Plots of $y_1$, $y_2$, and $y_3$ for $s=3$.}\label{fig:thresholds} 
\end{figure} 
  In what follows we provide bounds based on these functions to study pairwise stable and efficient structures, and investigate the convergence of the Formation Dynamics when possible.

\section{Results on Formation of Disjoint Cliques}

 We first study the inner structure of each group in a pairwise stable network.  Throughout this paper, we assume that the dynamics does not start with an initial state containing any interconnection.  We define the invariant set of all subgraphs of disjoint cliques as $S= \Big \{ \bigcup \limits_{\gamma=1}^{m} \Edges_{\gamma} \ | \  \Edges_{\gamma} \subset \Edges_{ K_{s_{\gamma}}}  \}$ where $\Edges_{\gamma}$ indicates the inner-network of group $P_{\gamma}$.

\begin{theorem}[Formation of Cliques: Pairwise Stability, Efficiency, Convergence]\label{thm: Pairwise Stability of Cliques}
Consider $n$ individuals partitioned into groups $P_1,\dots,P_m$. Then, each one of these $m$ groups is a clique in the pairwise stable and in the efficient structure if and only if $c<y_3$.  Moreover, starting from any state in the invariant set $S$, each group $P_{\alpha}$ will form a $s_{\alpha}$-size clique along Formation Dynamics (introduced in Section~\ref{sec:model}).
\end{theorem}
\begin{proof}
We first provide the proof of sufficiency for pairwise stability: for any individual $i\in P_{\alpha}$, a direct link with individual $j$, $(j\neq i)$ from the same group provides a profit of $\delta-c$. Without a direct link, this profit is equal to $\delta^{d_{ij}}$ where $d_{ij}>1$ is the distance between $i$ and $j$ in $\Edges \setminus \{(i,j)\}$. Since $\delta-\delta^2>c$, we have 
  $\delta-c>\delta^2 >\dots >\delta^n$; meaning that all
  agents prefer direct links to any indirect link. Thus, if agents $i$ and
  $j$ in group $P_\alpha$ are not directly connected, they will form a link
  and each will gain at least $(\delta-c)-\delta^{d_{ij}}>0$,
  i.e.,
  \[
  \begin{aligned} 
  &\text{for all } \{(i,j)\} \notin \Edges,\quad i, j
    \in P_{\alpha}, \ i \neq j \\
   & \quad  U_i(\Edges)<U_i(\Edges\cup\{(i,j)\}),
    \text{ and } U_j(\Edges)<U_j(\Edges\cup\{(i,j)\}).
  \end{aligned}
  \]
  Moreover, no node has an incentive to break any link since its payoff strictly
  decreases if it do so, i.e., 
  \[
  \begin{aligned}
 &   \text{for all } \{(i,j)\} \in \Edges, \quad i, j \in P_{\alpha},\ i \neq j , \\
      &  \quad U_i(\Edges) > U_i(\Edges \setminus \{(i,j)\}), \text{ and }
    U_j(\Edges) > U_j(\Edges \setminus \{( i,j) \}).
  \end{aligned}
  \]
  Thus, each group forms a clique and no intra-connection is removed after being formed, and according to Lemma \ref{lem:stable-by-dynamics},  these $m$ groups are cliques in the pairwise stable structure.
To prove necessity,  assume we have a pairwise stable clique. 
For $P_\alpha$ to remain a clique, all pairs of nodes belonging to the same group should prefer to keep one-hop links rather than having links with larger lengths, and thus $\delta-c>\delta^2>\delta^3> \dots$. This proves the claim that each group $P_\alpha$ is a clique if and only if $c<\delta-\delta^2$. Convergence of dynamics to cliques can be obtained directly from the same argument.

 We now continue by first proving that if $c<\delta - \delta^2$, in the efficient structure each group is a clique. From the analysis above, when $c<\delta-\delta^2$, we have:
\[
\begin{aligned}
v\big(\Edges\cup\{(i,j)\}\big)&-v\big(\Edges\setminus\{(i,j)\}\big) \\
  & \geq U_i\big(\Edges\cup\{(i,j)\}\big)+U_j\big(\Edges\cup\{(i,j)\}\big)
    \\
    &\quad-U_i\big(\Edges\setminus\{(i,j)\}\big) -U_j\big(\Edges\setminus\{(i,j)\}\big)\\
 &   \geq 2(\delta -c-\delta^{2})>0
\end{aligned}
\] 
which holds for each pair $(i,j)$ belonging to the same group, meaning that each group is a clique in the efficient structure.  
We next prove necessity for efficiency: assume $\Edges$ is the efficient structure and each group is a clique, i.e., $\{(i,j)\}\in \Edges$ for any two individuals $i, j$, $(i \neq j)$ from the same group. Then, we have:
\[
\begin{aligned}
v(\Edges)&-v(\Edges\setminus\{(i,j)\})\\
 & =U_i(\Edges)+U_j(\Edges)-U_i(\Edges\setminus\{(i,j)\})-U_j(\Edges\setminus\{(i,j)\})\\
 &=2(\delta -c-\delta^{2})>0,
\end{aligned}
\]
which results in $c<y_3$.
\end{proof}

\textbf{Note}: Theorem~\ref{thm: Pairwise Stability of Cliques} implies that formation of cliques requires $c < 1/4$.

\begin{theorem}[Pairwise Stable Structures and Convergence: Disjoint Union of Cliques]\label{thm:notequal-Disjoint} 
Consider $n$ individuals partitioned into groups $P_1,\dots,P_m$ of sizes $s_{1},\dots,s_{m}$ respectively. Assume that $c<y_3$. Then, the unique pairwise stable structures consists of disjoint union of cliques equal to the groups $P_1,\dots,P_m$ if and only if $F_{\alpha \beta} \leq \max\limits_{s \in\{s_{\alpha}, s_{\beta}\}} \dfrac{c}{ y_1(s)} $ for all $\alpha, \beta \in \until{m}$, $\alpha \neq \beta$. Moreover, starting from any state in the invariant set $S$, Formation Dynamics (introduced in Section~\ref{sec:model}) converges to this pairwise stable structure.
\end{theorem}

\begin{proof}
We first prove sufficiency.  Since no interconnection exists in the invariant set $S$, for any network $\Edges\in S$, suppose two individuals $i\in P_{\alpha}$ and $j\in P_{\beta}$ are picked to decide whether to add the corresponding interconnection or not. We know that $U_{i}(\Edges \cup \{(i,j)\})\leq F_{\alpha \beta}y_{1}(s_{\beta})-c$ and 
$U_{j}(\Edges \cup \{(i,j)\})\leq F_{\alpha \beta}y_{1}(s_{\alpha})-c$, 
where equalities hold when both groups form cliques. Since $F_{\alpha \beta} \leq \max\limits_{s \in\{s_{\alpha}, s_{\beta}\}} \dfrac{c}{ y_1(s)} $, at least one of $U_{i}\big(\Edges \cup \{(i,j)\}\big)\leq 0$ and $U_{i}\big(\Edges \cup \{(i,j)\}\big)\leq 0$ holds. Therefore, the interconnection 
$\{(i,j)\}$ does not belong to pairwise stable structure and it does not form. From Theorem \ref{thm: Pairwise Stability of Cliques}, we know that unique stable state consists of the disjoint union of cliques equal to the groups $\until{m}$. Suppose that all groups form cliques at a time $t^{*}$. From then on, no link will be added or removed. 
According to Lemma \ref{lem:stable-by-dynamics},  the pairwise stable structure consists of the disjoint union of cliques equal to the groups $\until{m}$, and the network converges to this unique stable state.

To prove necessity, take any two individuals $i\in P_{\alpha}$ and $j\in P_{\beta}$. Since $\{(i,j)\}$ does not belong to pairwise stable structure, we have at least one of 
$U_{i}(\Edges \cup \{(i,j)\})\leq 0$ and $U_{i}(\Edges \cup \{(i,j)\})\leq 0$ holds, and therefore, $F_{\alpha \beta} \leq \max\limits_{s \in\{s_{\alpha}, s_{\beta}\}} \dfrac{c}{ y_1(s)} $. 

\end{proof}

\section{Two Group Connectivity Structure}

In this section we study pairwise stable and efficient structures when individuals are partitioned into two groups. 

\subsection{Pairwise Stability}
In what follows we give the sufficient and necessary condition for pairwise stable structures.

\begin{theorem}[Pairwise Stability and Convergence with Two groups]\label{thm:existence-convergence-two}
  Consider $n$ individuals partitioned into groups $P_1, P_2$ of sizes  $s_1$ and $s_2$ respectively. Then, under the assumption $c<y_3$, the network has 
  \begin{enumerate}
   \item\label{fact:bridge-2groups} a unique
     pairwise stable structure consisting of minimally connected cliques if and
     only if  
     \[ 
     \max\limits_{s \in\{s_{1}, s_{2}\}}  \dfrac{c}{ y_1(s)}  \leq   
     F_{12} < \max\limits_{s \in\{s_{1}, s_{2}\}} \dfrac{c}{  y_2(s)}  ;
     \]  
   \item\label{fact:redundant-2groups} a unique pairwise stable structure consisting of exact $k~ (2\leq k\leq \min\{s_{1},s_{2}\})$  interconnections if and only if 
     \begin{multline}
    \max\limits_{s \in\{s_{1}, s_{2}\}} \dfrac{c}{ y_2(s)-(k-2)\delta y_{3}}  \\
    < F_{12} 
      < \max\limits_{s \in\{s_{1}, s_{2}\}} \dfrac{c}{ y_2(s)-(k-1)\delta y_{3}}.
     \end{multline}
   \item\label{fact:comember-2groups} a unique
     pairwise stable structure consisting of maximally connected cliques if and
     only if 
     \[\dfrac{c}{ y_3 } < F_{12} \leq 1.\]      
  \end{enumerate} 
  Moreover, starting from any state in the invariant set $S$, the Formation Dynamics (introduced in Section~\ref{sec:model}) converges to the corresponding pairwise stable structure. 
 \end{theorem}
\begin{figure}[ht]
	\includegraphics[width=1\linewidth]{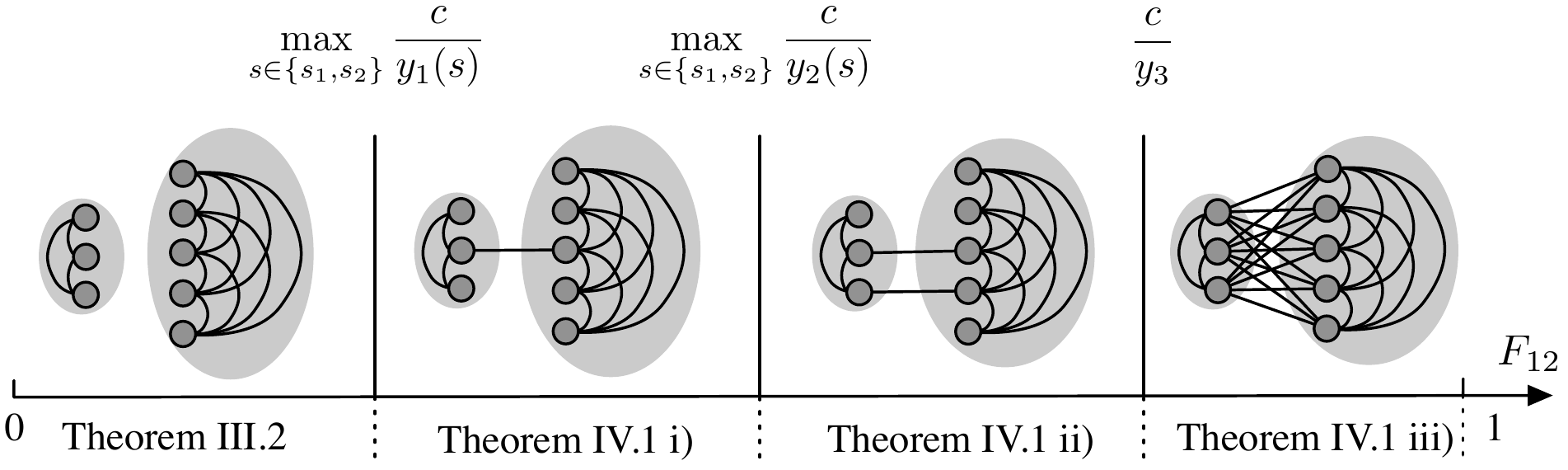}
        \caption{An illustration of the four ranges of parameter space in Theorems ~\ref{thm:notequal-Disjoint} and~\ref{thm:existence-convergence-two}.} \label{fig:spectrum}  
\end{figure} 
\begin{proof}

From Theorem \ref{thm:notequal-Disjoint}, we know that at least one interconnection will be formed under Formation Dynamics if and only if $ F_{12} \geq \max\limits_{s \in\{s_{1}, s_{2}\}} \dfrac{c}{ y_1(s)} $. Without loss of generality, suppose that $\Edges$ contain two cliques and at least one interconnection between them. Assume that agents $i$ and $j$, respectively from $P_{1}$ and $P_{2}$, are connected in $\Edges$. 

Take  agents $\hat{i} $ from $P_{1}$ and $\hat j$ from $P_{2}$. For $\hat{i} \neq i, \hat j=j$, we have
$
  U_{\hat i}(\Edges \cup \{ (\hat{i},j) \} )=  (s_{1} -1) \delta+ F_{12} \delta + (s_{2} -1) F_{12} \delta^2 - s_{1} c
$
and 
$
 U_{\hat{i}}(\Edges \setminus \{ (\hat{i},j) \} )=(s_{1} -1) (\delta- c) +  F_{12} \delta^2 + (s_{2} -1) F_{12} \delta^3,
$
implying $ U_{\hat i}(\Edges \cup \{ (\hat{i},j) \} )> U_{\hat{i}}(\Edges \setminus \{ (\hat{i},j)\} )   
   \iff F_{12} > \dfrac{c}{y_2(s_2)}.$
   From 
$
 U_{j}(\Edges \cup \{ (\hat{i},j)\} )=  (s_{2} -1) \delta +2 F_{12} \delta + (s_{2} -2) F_{12} \delta^2 +(s_{2} +1) c
$
and 
$
 U_{j}(\Edges \setminus \{ (\hat{i},j )\} )=  (s_{2} -1) \delta + F_{12} \delta + (s_{2} -1) F_{12} \delta^2 +s_{2} c
$,
we obtain:  
$ U_j(\Edges \cup \{ (\hat{i},j) \} ) > U_j(\Edges \setminus \{ (\hat{i},j)\} )   
   \iff F_{12} > \dfrac{c}{y_{3}}.
$
Then, from $y_2(s)>y_{3}>0$, we conclude that an additional interconnection $\{(\hat i, j)\}$ is added and maintained if and only if  $F_{12} > \dfrac{c}{y_{3}}$. 
Similarly, an additional interconnection $\{( i,\hat j)\}$ is added and maintained if and only if  $F_{12}> \dfrac{c}{y_{3}}$. 
For $\hat i \neq i, \hat j\neq j$, using a similar argument, we obtain that $U_{\hat i}(\Edges \cup \{ (\hat{i},\hat j ) \} )> U_{\hat{i}}(\Edges \setminus \{ (\hat{i},\hat j )\} )   
   \iff F_{12} > \dfrac{c}{y_2(s_2)}$ and 
   $U_{\hat j} (\Edges \cup \{ (\hat{i},\hat j) \} )> U_{\hat j} (\Edges \setminus \{( \hat{i},\hat j )\} )   
   \iff F_{12}> \dfrac{c}{y_2(s_1)}$; which means that an additional interconnection $\{(\hat i, \hat j)\} $ is added and maintained if and only if $ F_{12} > \max\limits_{s \in\{s_{1}, s_{2}\}} \dfrac{c}{ y_2(s)} $ (strictly holds when $s_{1} = s_{2}$). Thus, we conclude that at least two interconnections are added and maintained if and only if $F_{12} > \min\left\{\dfrac{c}{y_{3}},\max\limits_{s \in\{s_{1}, s_{2}\}} \dfrac{c}{ y_2(s)}\right\}=\max\limits_{s \in\{s_{1}, s_{2}\}} \dfrac{c}{ y_2(s)}$  (strictly holds when $s_{1} = s_{2}$.)  
Therefore, the network contains precisely one interconnection if and only if $ \max\limits_{s \in\{s_{1}, s_{2}\}}  \dfrac{c}{ y_1(s)}  <  
     F_{12} < \max\limits_{s \in\{s_{1}, s_{2}\}} \dfrac{c}{  y_2(s)} $. Moreover, from the moment when two group form cliques and this unique interconnection builds, the network will not change.  According to Lemma \ref{lem:stable-by-dynamics}, this concludes the proof of statement~\ref{fact:bridge-2groups}. 

To prove~\ref{fact:comember-2groups}, assume that $F_{12}>  \max\limits_{s \in\{s_{1}, s_{2}\}} \dfrac{c}{ y_2(s)} $. We have shown above that $\Edges$ contains at least two interconnections between two cliques. As a result, for any agent $\hat i$ from $P_{1}$ and $\hat j$ from $P_{2}$, the distance between $\hat i$ and $\hat j$ in $\Edges \setminus \{(\hat i, \hat j)\}$ is equal to either 2 or 3.  If it is equal to 2, then  $U_{\hat i}(\Edges \cup \{ (\hat{i}, \hat j ) \} -U_{\hat{i}}(\Edges \setminus \{ (\hat{i}, \hat j )\}=U_{\hat j}(\Edges \cup \{( \hat{i}, \hat j )\} -U_{\hat{j}}(\Edges \setminus \{ (\hat{i}, \hat j )\}=F_{12}(\delta-\delta^{2})-c$; and if it is equal to 3,  then $U_{\hat i}(\Edges \cup \{( \hat{i}, \hat j  )\} -U_{\hat{i}}(\Edges \setminus \{ (\hat{i}, \hat j) \}=U_{\hat j}(\Edges \cup \{ (\hat{i}, \hat j )\} -U_{\hat{j}}(\Edges \setminus \{ (\hat{i}, \hat j )\}=F_{12}(\delta-\delta^{3})-c$. Interconnection $\{(\hat i, \hat j)\}$ is added and maintained if and only if $U_{\hat i}(\Edges \cup \{ (\hat{i},\hat j ) \} -U_{\hat{i}}(\Edges \setminus \{ (\hat{i}, \hat j) \}> 0$ and $U_{\hat j}(\Edges \cup \{( \hat{i}, \hat j) \} -U_{\hat{j}}(\Edges \setminus \{ (\hat{i}, \hat j )\}>0$. 
Thus, we conclude that $\{(\hat i, \hat j)\}$ is added and maintained if and only if $F_{12}>\max\left\{\dfrac{c}{\delta-\delta^{2}},\dfrac{c}{\delta-\delta^{3}}\right\}$. Since $\max\left\{\dfrac{c}{\delta-\delta^{2}},\dfrac{c}{\delta-\delta^{3}}\right\}=\dfrac{c}{\delta-\delta^{2}}$ for $0<\delta<1$,  $\{(\hat i, \hat j)\}$ is added and maintained maintained if and only if $F_{12}>\dfrac{c}{\delta-\delta^{2}}=\dfrac{c}{y_{3}}$. Therefore, the network will not be changed when all agents link with each other. By Lemma \ref{lem:stable-by-dynamics}, this concludes the proof of~\ref{fact:comember-2groups}.

From statements \ref{fact:bridge-2groups} and \ref{fact:comember-2groups}, we know that the pairwise stable structure contains at least 2 but not fully numbers of interconnections if $\max\limits_{s \in\{s_{1}, s_{2}\}} \dfrac{c}{ y_2(s)} < F_{12} < \dfrac{c}{ y_3 }$.  Suppose that $(i_{1},j_{1}), \dots, (i_{k-1},j_{k-1})$ are $k-1$ interconnections between $P_{1}$ and $P_{2}$. Take  agents $\hat{i} $ from $P_{1}$ and $\hat j$ from $P_{2}$. Similar to the analysis in the proof of statement~\ref{fact:bridge-2groups}, we have the following two cases:
\begin{enumerate}[(a)]
\item For $\hat{i} \notin \{i_{1},\dots,i_{k-1}\}, \hat j\notin \{j_{1},\dots, j_{k-1}\}$, we have
\[
\begin{aligned}
   U_{\hat i}(\Edges \cup \{ (\hat{i},j) \} )&- U_{\hat{i}}(\Edges \setminus \{ (\hat{i},j) \} )\\
 &\quad = F_{12}(y_{2}(s_{2})-(k-2)\delta y_{3}), \text{ and}\\
    U_{\hat j}(\Edges \cup \{ (\hat{i},j) \} ) &- U_{\hat{i}}(\Edges \setminus \{ (\hat{i},j) \} )\\
&\quad  = F_{12}(y_{2}(s_{1})-(k-2)\delta y_{3}),
\end{aligned}
\]
implying 
\[\begin{aligned}
U_{\hat i}(\Edges \cup \{ (\hat{i},j) \} )&>U_{\hat{i}}(\Edges \setminus \{ (\hat{i},j)\} ), \text {~and~}  \\
U_{\hat j}(\Edges \cup \{ (\hat{i},j) \} )&> U_{\hat{j}}(\Edges \setminus \{ (\hat{i},j)\} ) \\
  \iff \qquad F_{12} &> \max\limits_{s \in\{s_{1}, s_{2}\}}  \dfrac{c}{y_2(s)-(k-2)\delta y_{3}}.
   \end{aligned}\]
\item For $\hat{i} \in \{i_{1},\dots,i_{k-1}\}, \hat j\notin \{j_{1},\dots, j_{k-1}\}$ or  $\hat{i} \notin \{i_{1},\dots,i_{k-1}\}, \hat j\in \{j_{1},\dots, j_{k-1}\}$, we have 
\[\begin{aligned}
 U_{\hat i}(\Edges \cup \{ (\hat{i},j) \} )&> U_{\hat{i}}(\Edges \setminus \{ (\hat{i},j)\} ) 
   \text {~ and~}\\
 U_{\hat j}(\Edges \cup \{ (\hat{i},j) \} )&>U_{\hat{j}}(\Edges \setminus \{ (\hat{i},j)\} ) \\
    \iff \qquad F_{12} &>   \dfrac{c}{y_{3}}.   
   \end{aligned}\]
\end{enumerate}
Therefore,  we conclude then $k-th$ interconnection is added and maintained if and only if  $F_{12} >\max\limits_{s \in\{s_{1}, s_{2}\}}  \dfrac{c}{y_2(s)-(k-2)\delta y_{3}}.$ Likewise, the $(k+1)-th$ interconnection is added and maintained if and only if  $F_{12} > \max\limits_{s \in\{s_{1}, s_{2}\}}  \dfrac{c}{y_2(s)-(k-1)\delta y_{3}}.$ It follows that the unique pair-wise stable structure has exact $k~ (2\leq k\leq \min\{s_{1},s_{2}\})$  interconnections if and only if $\max\limits_{s \in\{s_{1}, s_{2}\}} \dfrac{c}{ y_2(s)-(k-2)\delta y_{3}}< F_{12} < \max\limits_{s \in\{s_{1}, s_{2}\}} \dfrac{c}{ y_2(s)-(k-1)\delta y_{3}}$. 
This concludes the proof of~\ref{fact:redundant-2groups}.

 Finally, we complete the proof of Theorem~\ref{thm:existence-convergence-two} by proving the convergence statement. Since  $c<y_3$, from Theorem \ref{thm: Pairwise Stability of Cliques} we know that the network structure, regardless of the density of intra-group connections, consists of cliques of sizes $s_1$ and $s_2$.  From above analysis,  starting from invariable set $S$, intra-connections will increase one by one until one more intra-connection can not bring increasing of benefit for two players. Then, it follows from Lemma \ref{lem:stable-by-dynamics} that the network can not be changed from then on, i.e., the network will converge to the corresponding pairwise stable structure under the formation dynamics in Section~\ref{sec:model}. 
\end{proof}  

Theorem \ref{thm:existence-convergence-two} implies that if $F_{12}$ equals the boundaries, i.e., 
$
\max\limits_{s \in\{s_{1}, s_{2}\}} \dfrac{c}{y_2(s)}, \max\limits_{s \in\{s_{1}, s_{2}\}} \dfrac{c}{y_2(s)-\delta y_{3}}, \max\limits_{s \in\{s_{1}, s_{2}\}} \dfrac{c}{ y_2(s)-2\delta y_{3}}$, $\dots, \max\limits_{s \in\{s_{1}, s_{2}\}} \dfrac{c}{y_{3}},
$
 the pairwise stable structure is not unique. To show this, we provide the following example for a society consisting of 8 individuals.

\begin{example}{\label{example:non-unique-pwStable}}
 Suppose that individuals are partitioned into groups $P_{1}=\{1,2,3\}$ and $P_{2}=\{4,5,6,7,8\}$. Let $c=0.2$, $\delta=0.5$ and  $F_{12}=\max\limits_{s \in\{s_{1}, s_{2}\}} \dfrac{c}{y_2(s)}=0.4$. Assume at time $0$, each of the two groups have a line structure, illustrated in Fig~\ref{fig:stabilityatBoundary}. Now consider the following two processes:
\begin{enumerate}[Process A.]
\item At first, individuals 2 and 6 are chosen to play the game and interconnection $(2,6)$ is formed. Then, individuals 3 and 7 are chosen to play the game, and interconnection $(3,7)$ is also formed.  After that, all possible intra-connections are considered and formed. Finally, the network contains two interconnections and reaches pairwise stability.
\item At first, all possible intra-connections are considered and formed. Then individuals 2 and 6 are chosen to play the game and interconnection $(2,6)$ is formed. As a result, the network contains only 1 interconnection and reaches pairwise stability.
\end{enumerate}

\end{example}
\begin{figure}[ht]
	\centering
	\includegraphics[width=1\linewidth]{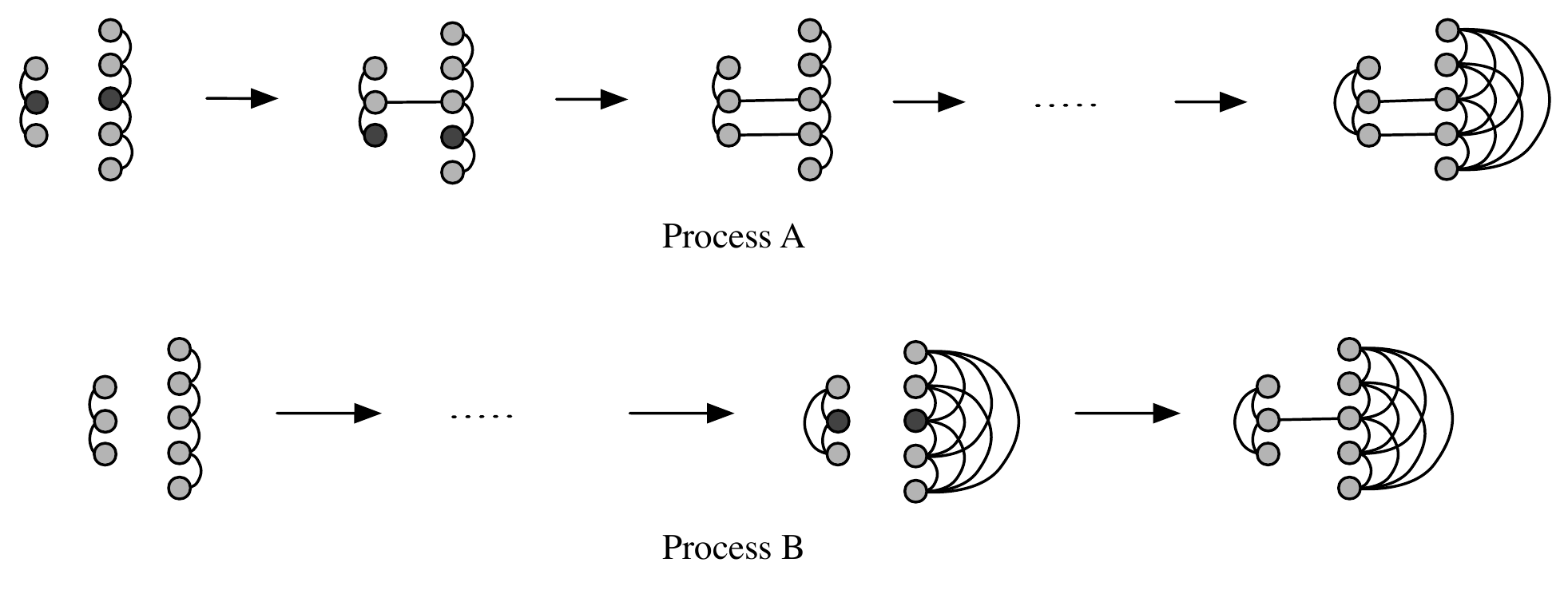}
        \caption{An illustration of two processes in Example~\ref{example:non-unique-pwStable}. At each step, the darker dots are chosen to play the game.} \label{fig:stabilityatBoundary}  
       \end{figure}

Fig.~\ref{fig:spectrum} illustrates the scenarios specified in Theorems ~\ref{thm:notequal-Disjoint} and~\ref{thm:existence-convergence-two} for $F_{12}$.

\begin{remark}{}

As the difference between group sizes $|s_2- s_1|$ increases, the bound for formation of bridges increases, and thus communication between groups becomes harder.
To see why this is true, for a society of $n=s_1+s_2$ individuals, without loss of generality, we assume that $s_1 \leq s_2.$.
We know $\max\limits_{s \in\{s_{1}, s_{2}\}} \dfrac{c}{ y_1(s)}= \dfrac{c}{ \min\limits_{s \in\{s_{1}, s_{2}\}} y_1(s)}= \dfrac{c}{ y_1(\min\{s_{1}, s_{2}\})}$, and since $\max\limits_{s \in\{s_{1}, s_{2}\}} \dfrac{c}{ y_1(s)}$ is a monotonically decreasing function of $s$, as $s_1$ increases, $|s_2-s_1|$, and therefore, $\max\limits_{s \in\{s_{1}, s_{2}\}} \dfrac{c}{ y_1(s)}$ decreases, and communication is facilitated. As illustrated in Fig.~\ref{fig:2x2_var}, for a society of fixed size, as the sizes of the two groups becomes closer to each other, the number of interconnections increases.
\end{remark}
\begin{figure}
	\begin{center} 
		%\subfloat[Number of Interconnections]
		{\includegraphics[width=0.9\linewidth]{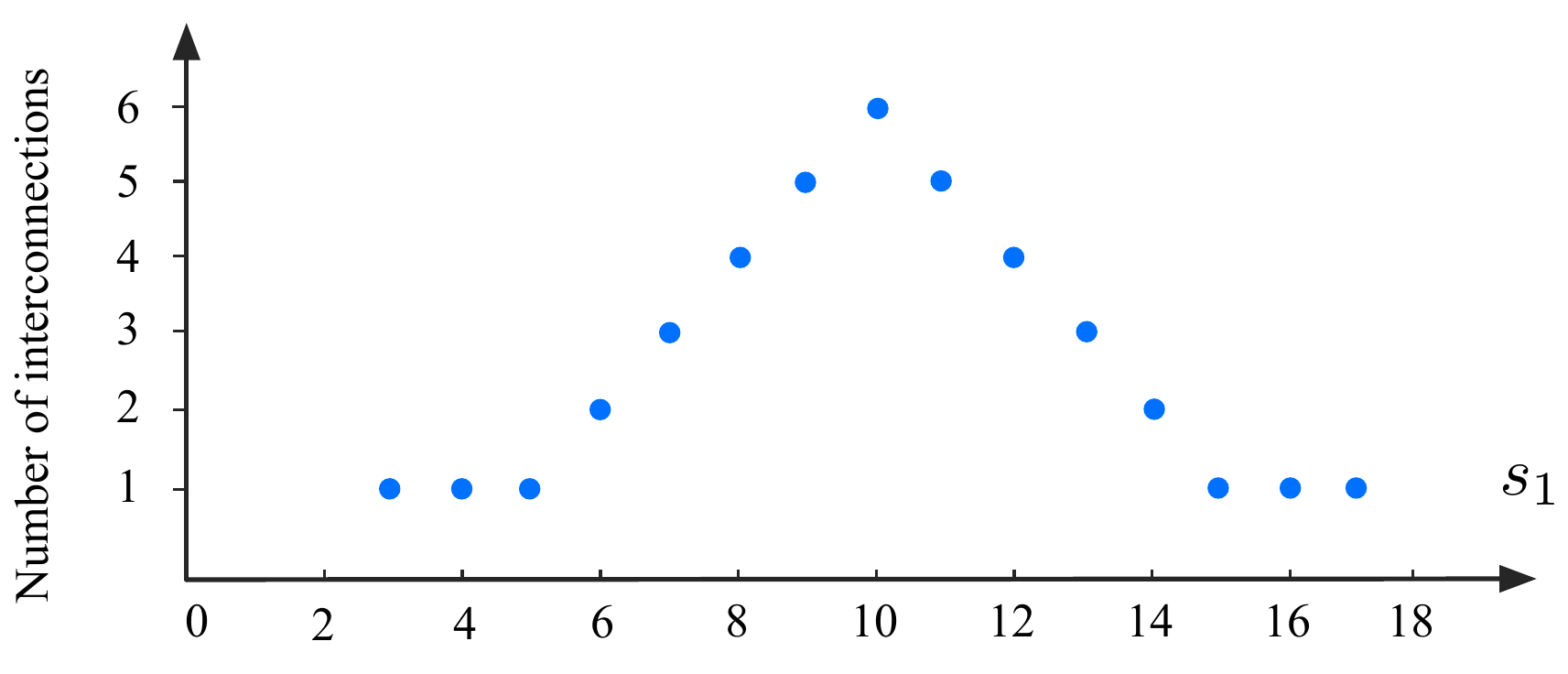}\label{fig:2x2-interC}}
		%\subfloat[Social Welfare]{\includegraphics[width=0.475\linewidth]{value_2x2_mesh}\label{fig:2x2-value}} 
		\caption{\small Plots of no. of interconnections with fixed $F_{12}=0.25$ and $n=20$ as a function of groups size $s_1$, $\delta=0.5$, $c=0.2$. Note that $s_2=n-s_1$, and thus the plot is symmetric about the line $s_1=\frac{n}{2}$ and maximized at $s_1=\frac{n}{2}$.}\label{fig:2x2_var}
	\end{center}
\end{figure}

\subsection{Efficiency}

We now study the efficiency of two group structure.

\begin{theorem}[Efficiency with Two groups]\label{thm:efficiency-two}
   Consider $n$ individuals partitioned into groups $P_1, P_2$ of sizes  $s_1$ and $s_2$ respectively. Under the assumption that $c<y_3$, the efficient structure consists of:
  \begin{enumerate}
  \item\label{fact:eff-disjoint-2groups} 
     disjoint union of cliques if and only if  
             \[ 0 \leq F_{12}  \leq  \dfrac{c\delta}{y_{1}(s_{1})y_1(s_{2}) } ;\]
  \item\label{fact:eff-bridge-2groups} 
       minimally connected cliques if and only if  
         \[
          \dfrac{c\delta}{y_{1}(s_{1})y_1(s_{2}) } 
            \leq   F_{12} \leq \dfrac{2c}{ y_{2}(s_{2})+y_{2}(s_{1})+(s_{1}+s_{2}-4)\delta y_3};
            \]
  \item\label{fact:eff-redundant-2groups} 
       redundantly connected cliques if  and only if 
               \[ \dfrac{2c}{ y_{2}(s_{2})+y_{2}(s_{1})+(s_{1}+s_{2}-4)\delta y_{3}}
                \leq F_{12} \leq \dfrac{c}{ y_3 };\]
   \item\label{fact:eff-comember-2groups} 
        maximally connected cliques if and   only if 
        \[\dfrac{c}{ y_3 }\leq F_{12} \leq 1.\]      
  \end{enumerate}  
\end{theorem}
\begin{proof}
From Theorem \ref{thm: Pairwise Stability of Cliques}, we know that addition of an intra-connection results in increasing the payoff of both individuals involved in that link, as well as the social welfare, and removal of any intra-connection causes loss for both individuals and decreases social welfare. 
As a result, the efficient structure consists of two cliques. Suppose $i\in P_{1}$ and $j\in P_{2}$. Let $\Edges_{0}=\Edges_{K_{s_{1}}}\bigcup \Edges_{K_{s_{2}}}$ be the union of two cliques of sizes $s_1$ and $s_2$. We have 
\begin{equation}
\begin{aligned}
 U_{k}(\Edges_{0}\cup&\{(i,j)\}) -U_{k}(\Edges_{0}) 
\\
&= \left\{
\begin{aligned}
& F_{12}\big(\delta+(s_{2}-1)\delta^{2}\big)-c, && k=i;\\
& F_{12}\big(\delta+(s_{1}-1)\delta^{2}\big)-c, && k=j;\\
& F_{12}\big(\delta^{2}+(s_{2}-1)\delta^{3}\big), && k\in P_{1}, k\neq i;\\
& F_{12}\big(\delta^{2}+(s_{1}-1)\delta^{3}\big), && k\in P_{2}, k\neq j.
\end{aligned}
\right.
\end{aligned}
\end{equation}
It follows that 
\begin{multline*}
v\big((\Edges_{0}\cup \{(i,j)\}\big) -v(\Edges_{0})
=2F_{12}[\delta+(s_{1}-1)\delta^{2}+(s_{2}-1)\delta^{2}\\
+(s_{1}-1)(s_{2}-1)\delta^{3}]-2c.
\end{multline*}
$\Edges_{0}$ is the efficient structure if and only if $v\big((\Edges_{0}\cup \{(i,j)\}\big) -v(\Edges_{0})\leq 0$, which is equivalent to:  
\begin{equation*}\label{eqn:F12-efficent-bound-i}
\begin{aligned}
F_{12} &\leq   \dfrac{c}{ \delta+(s_{1}-1)\delta^{2}+(s_{2}-1)\delta^{2} +(s_{1}-1)(s_{2}-1)\delta^{3}} \\
&= \dfrac{c \delta}{ y_1(s_{1})y_1(s_{2})}.
\end{aligned}\end{equation*}
This concludes the proof of~\ref{fact:eff-disjoint-2groups}.

Now let $\Edges_{1}=\Edges_{0}\bigcup \{(i,j)\}$.  As elaborated above,  $v(\Edges_{1}) -v(\Edges_{0})\geq 0$ if and only if 
\[\begin{aligned}
F_{12} &\geq   \dfrac{c}{\delta+(s_{1}-1)\delta^{2}+(s_{2}-1)\delta^{2} +(s_{1}-1)(s_{2}-1)\delta^{3}}\\
&=\dfrac{c\delta}{y_{1}(s_{1})y_1(s_{2}) } .
\end{aligned}\]
Suppose  $\hat i \in P_{1}$ and $\hat j\in P_{2}$. For the case of $\hat i \neq i$ and $\hat j\neq j$, we have 
\begin{equation*}
U_{k}(\Edges_{1}\cup \{(\hat i,\hat j)\})-U_{k}(\Edges_{1})=
 \left\{
\begin{aligned}
& F_{12}y_{2}(s_{2})-c, && k=\hat i \\
&F_{12}y_{2}(s_{1})-c, && k=\hat j\\
& 0, && k =  i, j \\
& F_{12}(\delta^{2}-\delta^{3}), && \text{otherwise.}
\end{aligned}
\right.
\end{equation*}
For the case of $\hat i =i$ and $\hat j\neq j$, we have 
\begin{equation*}
\begin{aligned}
U_{k}(\Edges_{1}\cup&\{( i,\hat j)\})-U_{k}(\Edges_{1})\\
&= \left\{
\begin{aligned}
& F_{12}(\delta-\delta^{2})-c, && k=	i, \\
&F_{12}y_{2}(s_{1})-c, && k=\hat j,\\
& 0, && k\in P_{2} ~\text{and}~ k\neq \hat j ,\\
& F_{12}(\delta^{2}-\delta^{3}), && k\in P_{1} ~\text{and}~ k\neq i.
\end{aligned}
\right.
\end{aligned}
\end{equation*}
For the case of $\hat i \neq i$ and $\hat j=j$, we have 
\begin{equation*}
\begin{aligned}
U_{k}(\Edges_{1}\cup&\{(\hat i, j)\})-U_{k}(\Edges_{1})\\
&= \left\{
\begin{aligned}
& F_{12}(\delta-\delta^{2})-c, && k=	j, \\
&F_{12}y_{2}(s_{2})-c, && k=\hat i,\\
& 0, && k\in P_{1} ~\text{and}~ k\neq \hat i,\\
& F_{12}(\delta^{2}-\delta^{3}), && k\in P_{2} ~\text{and}~ k\neq j.
\end{aligned}
\right.
\end{aligned}
\end{equation*}
It follows that 
\begin{equation*}
\begin{aligned}
&v(\Edges_{1}\cup \{(\hat i,\hat j)\})-v(\Edges_{1})\\
&=\left\{
\begin{aligned}
& F_{12}[y_{2}(s_{2})+y_{2}(s_{1})+(s_{1}+s_{2}-4)(\delta^{2}-\delta^{3})]-2c, \\
&\quad  \quad \quad \quad \quad\quad \quad     \hat j \neq j, \hat i \neq i,\\
&2F_{12}y_{2}(s_{1})-2c , \hat i  =  i,\hat j \neq j,\\
&2F_{12}y_{2}(s_{2})-2c ,  \hat i  \neq i,\hat j = j.\\
\end{aligned}
\right.
\end{aligned}
\end{equation*}

Since $y_2(s)$ is a monotonically increasing function of $s$ and  
\[ \begin{aligned}
&\dfrac{2c}{ y_{2}(s_{2})+y_{2}(s_{1})+(s_{1}+s_{2}-4)(\delta^{2}-\delta^{3})}\\
&\qquad=\dfrac{c}{y_{2}(s_{1})+(s_{2}-2)(\delta^{2}-\delta^{3})}\\
&\qquad=\dfrac{c}{y_{2}(s_{2})+(s_{1}-2)(\delta^{2}-\delta^{3})},
\end{aligned}\]
 we have
\begin{multline*}
\dfrac{2c}{ y_{2}(s_{2})+y_{2}(s_{1})+(s_{1}+s_{2}-4)(\delta^{2}-\delta^{3})}\\
< \min\limits_{s \in\{s_{1}, s_{2}\}}\dfrac{c}{y_{2}(s)}.
\end{multline*}
We then conclude that $v(\Edges_{2}\cup \{(\hat i,\hat j)\})-v(\Edges_{2})\leq 0$ if and only if $F_{12} \leq \dfrac{2c}{ y_{2}(s_{2})+y_{2}(s_{1})+(s_{1}+s_{2}-4)\delta y_3}$.  Statement~\ref{fact:eff-bridge-2groups} follows accordingly.

Assume $F_{12}\geq  \max\limits_{s \in\{s_{1}, s_{2}\}} \dfrac{c}{ y_2(s)} $. From statement~\ref{fact:eff-bridge-2groups}, we know that the efficient structure has at least two interconnections. For any agent $\hat i$ from $P_{1}$ and $\hat j$ from $P_{2}$, as elaborated in the proof of Theorem \ref{thm:existence-convergence-two}, one has
\[
v(\Edges \cup \{ (\hat{i},\hat j ) \} -v(\Edges \setminus \{ (\hat{i},\hat j) \}\geq 2[F_{12}(\delta-\delta^{2})-c],
\]
which concludes statements \ref{fact:eff-redundant-2groups} and \ref{fact:eff-comember-2groups}.

\end{proof} 

\begin{figure*}
	\begin{center} 
		\subfloat[Ranges of parameter space for high value benefit ($\delta \geq  \max\limits_{s \in\{s_{1}, s_{2}\}} \dfrac{s-3}{n-3}$)]{\includegraphics[width=0.9\linewidth]{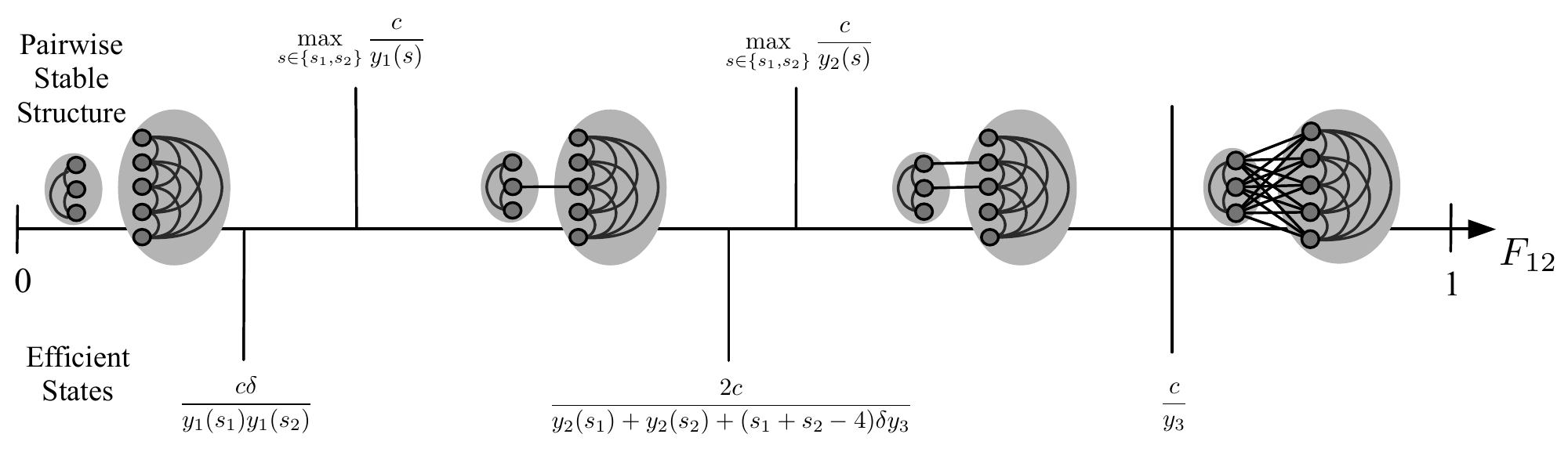}\label{fig:}}\\
		\subfloat[Ranges of parameter space for low value benefit ($\delta <  \max\limits_{s \in\{s_{1}, s_{2}\}} \dfrac{s-3}{n-3}$)]{\includegraphics[width=0.9\linewidth]{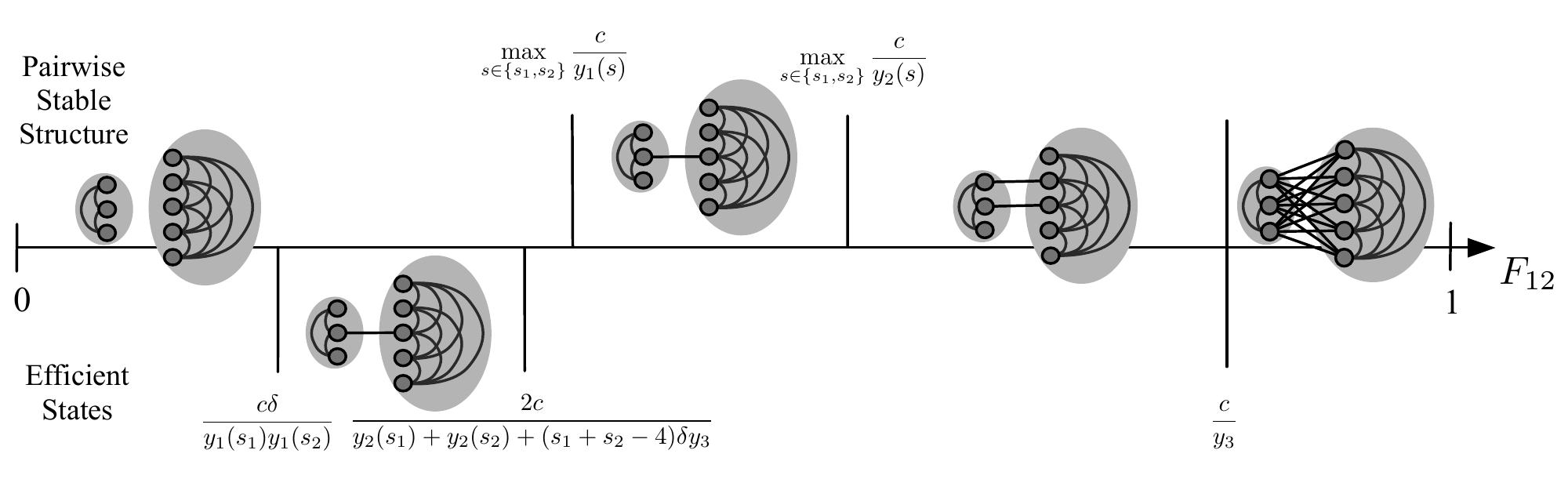}\label{fig:}}
		\caption{\small An illustration of the ranges of parameter space in Theorems~\ref{thm:existence-convergence-two} and~\ref{thm:efficiency-two}. The intervals indicated above the horizontal line refer to bounds of pairwise stability and those below the horizontal line refer to bounds of efficiency.}
		\label{fig:efficiency-stability-2groups}
	\end{center}
\end{figure*}

From the properties of $y_{1}(s)$, $y_{2}(s)$, and $y_{3}$, we have:
\begin{enumerate}
\item if $\delta \geq  \max\limits_{s \in\{s_{1}, s_{2}\}} \dfrac{s-3}{n-3}$, then
\[\begin{aligned}
 \dfrac{c\delta}{y_{1}(s_{1})y_1(s_{2}) }    
<&
 \max\limits_{s \in\{s_{1}, s_{2}\}} \dfrac{c}{ y_1(s)} \\
  \leq &
   \dfrac{2c}{ y_{2}(s_{2})+y_{2}(s_{1})+(s_{1}+s_{2}-4)\delta y_3} \\
 <&
  \max\limits_{s \in\{s_{1}, s_{2}\}} \dfrac{c}{  y_2(s)};
\end{aligned}
\]
\item  if $\delta <  \max\limits_{s \in\{s_{1}, s_{2}\}} \dfrac{s-3}{n-3}$, then
\[\begin{aligned}
 \dfrac{c\delta}{y_{1}(s_{1})y_1(s_{2}) }   & < \dfrac{2c}{ y_{2}(s_{2})+y_{2}(s_{1})+(s_{1}+s_{2}-4)\delta y_3} 
 \\&<
 \max\limits_{s \in\{s_{1}, s_{2}\}} \dfrac{c}{ y_1(s)}\\
 & <
  \max\limits_{s \in\{s_{1}, s_{2}\}} \dfrac{c}{  y_2(s)}.
\end{aligned}\]
\end{enumerate}
From Theorems \ref{thm:existence-convergence-two} and \ref{thm:efficiency-two}, we can directly obtain the conditions for the equivalence between pairwise stability and efficiency.

\begin{corollary}\label{the: equivalence-stability-efficiency-two}
Consider $n$ individuals partitioned into groups, $P_1$ and $P_2$  of sizes  $s_1$ and $s_2$ respectively. Under the assumption $c<y_3$, the efficient structure has equal or more interconnections than the pairwise stable structure. Moreover, the efficient structure is the pairwise stable structure if 
\begin{enumerate}
\item for $\delta \geq  \max\limits_{s \in\{s_{1}, s_{2}\}} \dfrac{s-3}{n-3},$
\[\begin{aligned}
&F_{12}\in \left[0, \dfrac{c\delta}{y_{1}(s_{1})y_1(s_{2}) } \right] \\
&\bigcup  \left[ \max\limits_{s \in\{s_{1}, s_{2}\}} \dfrac{c}{ y_1(s)} ,  \dfrac{2c}{ y_{2}(s_{2})+y_{2}(s_{1})+(s_{1}+s_{2}-4)\delta y_3}\right] \\
&\bigcup 
\left[ \dfrac{c}{  y_3},1\right], \text{ or}
\end{aligned}\] 
\item  for $ \delta <  \max\limits_{s \in\{s_{1}, s_{2}\}} \dfrac{s-3}{n-3}, $

$F_{12}\in \left[0,  \dfrac{c\delta}{y_{1}(s_{1})y_1(s_{2}) }  \right] \bigcup 
\left[\dfrac{c}{  y_3},1\right], $
\end{enumerate}

\end{corollary}
We illustrate the compatibility (and incompatibility) between the pairwise stable and efficient structures in Fig.~\ref{fig:efficiency-stability-2groups}. The intersections of efficient and pairwise stable structures are highlighted.

\begin{corollary}
Consider $n$ individuals partitioned into groups $P_1$ and $P_2$  of sizes  $s_1$ and $s_2$, respectively. Under the assumption $c<y_3$,  if the efficient structure is not pairwise stable, then the efficient structure has more interconnections than the pairwise stable structure.
\end{corollary}
\begin{proof}
If the efficient structure is not pairwise stable, it implies that it can be changed by Formation Dynamics. Suppose that it has less interconnections than pairwise stable structure. Then, adding interconnections can make it pairwise stable, which means that both players involved in the interconnection benefit from this interconnection. Then the social welfare increases which is conflict with the fact that the structure is efficient. Thus, the efficient structure has more interconnections than the pairwise stable structure.
\end{proof}

\begin{remark}\label{rm:PoA}
For the two intervals $\dfrac{c\delta}{y_{1}(s_{1})y_1(s_{2}) }   <F_{12}<\max\limits_{s \in\{s_{1}, s_{2}\}} \dfrac{c}{ y_1(s)}$ and $ \dfrac{2c}{ y_{2}(s_{2})+y_{2}(s_{1})+(s_{1}+s_{2}-4)\delta y_3}<F_{12}<  \max\limits_{s \in\{s_{1}, s_{2}\}} \dfrac{c}{  y_2(s)}$, the pairwise stable structure is not efficient. For $ \max\limits_{s \in\{s_{1}, s_{2}\}} \dfrac{c}{  y_2(s)}\geq F_{12}< \dfrac{c}{  y_3}$, the pairwise stable structure might be not efficient. This is because individuals are rational and selfish (i.e., they only care about their own payoffs.) Therefore, total social utility might drop when an individual tries to maximize its own payoff. Specifically, 
\begin{enumerate}
\item when $ \dfrac{c\delta}{y_{1}(s_{1})y_1(s_{2}) }  <F_{12}<\max\limits_{s \in\{s_{1}, s_{2}\}} \dfrac{c}{ y_1(s)}$,  an individual experiences loss by  having an interconnection, although that interconnection could bring profit for other individuals and, therefore, result in an increase of total utility. The individual refuses to {add} or maintain this interconnection.
\item when $ \dfrac{2c}{ y_{2}(s_{2})+y_{2}(s_{1})+(s_{1}+s_{2}-4)\delta y_3}<F_{12}< \dfrac{c}{  y_3}$,  as mentioned in the proof of Theorem \ref{thm:efficiency-two}, there exists at least one interconnection $\{(i,j)\}, i\in P_{1}, j\in P_{2}$. For the two individuals $\hat i\in P_{1}$ and $\hat j \in P_{2}$ ($\hat i\neq i, \hat j \neq j$), an interconnection causes loss for the player from the larger group, but brings profit for all other individuals, which results in removal of the interconnection or rejection in adding it and, therefore, in total utility loss. 
\end{enumerate}
\end{remark}

Fig.~\ref{fig:2x2} shows the number of interconnections and social welfare for both efficient and stable structures for two cliques of sizes 3 and 5, where the value of $\delta=0.5$ and $c=0.2$. Note that:
\begin{enumerate}
\item As shown in Fig.~\ref{fig:interC_2x2}, when $F_{12}=\dfrac{c}{y_3}$ which evaluates to 0.8 for the choice of our parameters, there exist non unique pairwise stable and efficient structures, and the number of interconnections for those structures vary between 3 and 15.
\item  As shown in Fig.~\ref{fig:interC_2x2} and Fig.~\ref{fig:value_2x2}, for $F_{12}\in (0.67, 0.2) \cup (0.23,0.53)$, the pairwise stable structure is not efficient and thus the social welfare of efficient structures is always greater than that of pairwise stable structures. Fig. \ref{fig:PoA} shows plot of PoA as a function of $F_{12}$. We observe that the slope for $F_{12}\in (0.67, 0.2) $ is the highest. This is because in this interval, the pairwise stable structure dose not have any interconnection, whereas the efficient structure has one interconnection. This link brings large value for the overall network and makes a large difference in social welfare of two kinds of structures. 
\end{enumerate}

 \begin{figure}
	\begin{center} 
		\subfloat[Number of Interconnections]{\includegraphics[width=0.99\linewidth]{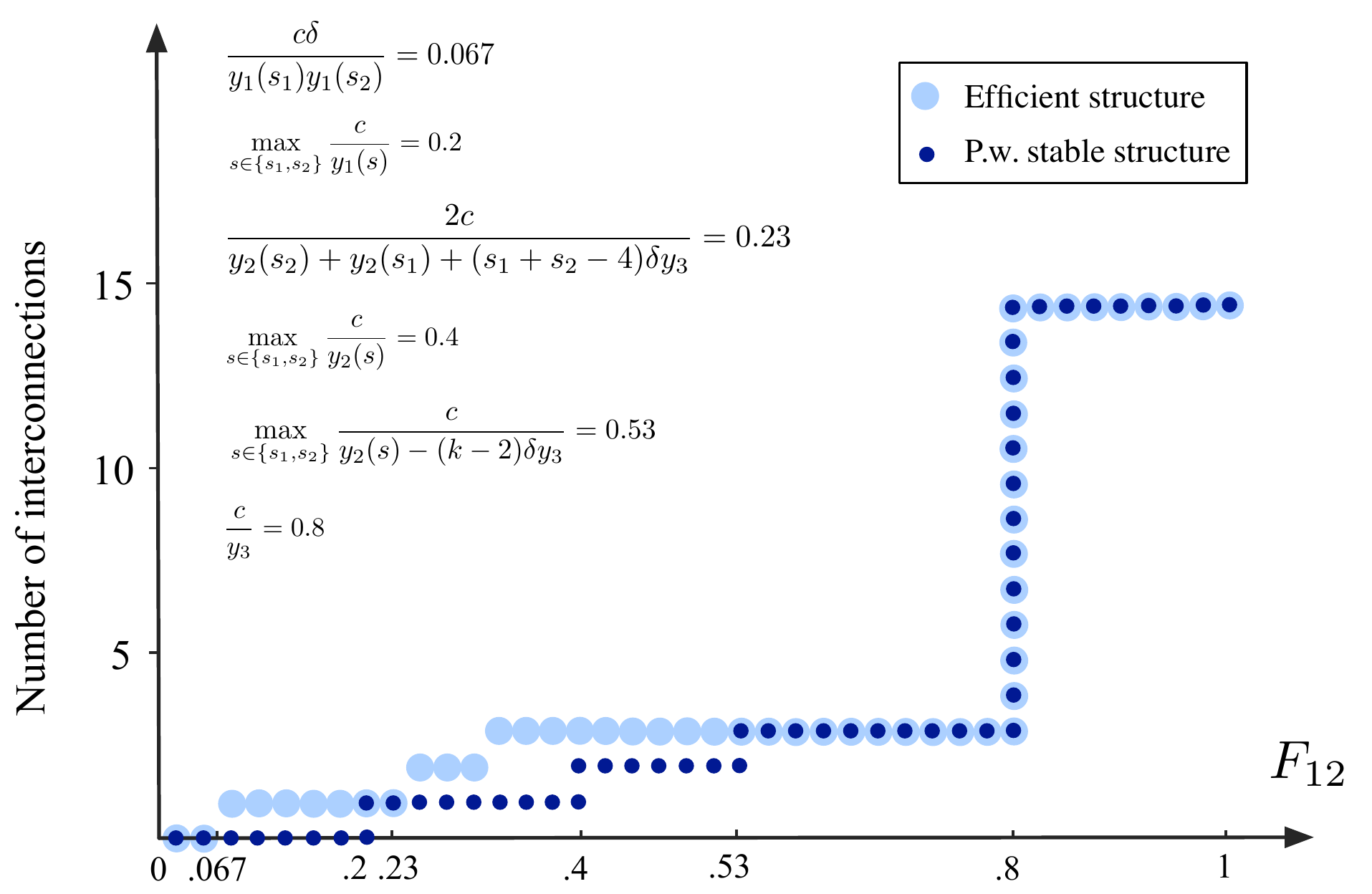}\label{fig:interC_2x2}}\\
		\subfloat[Social Welfare]{\includegraphics[width=0.95\linewidth]{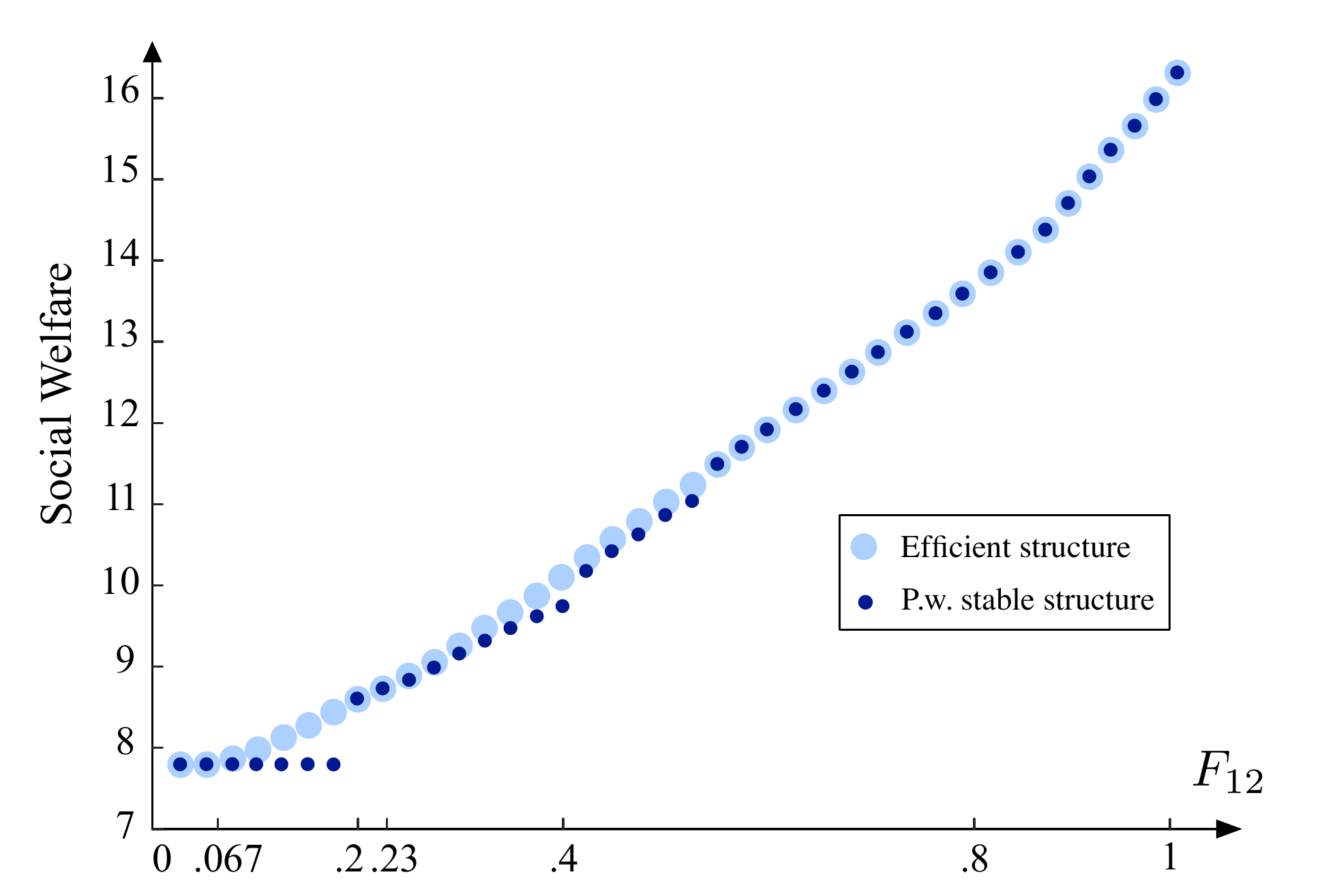}\label{fig:value_2x2}}
		\caption{\small Plots of no. of interconnections and social welfare as a function of $F_{12}$ for two groups of sizes 3 and 5, $\delta=0.5$, $c=0.2$.}\label{fig:2x2}
	\end{center}
\end{figure}
\begin{figure}
	{\includegraphics[width=0.9\linewidth]{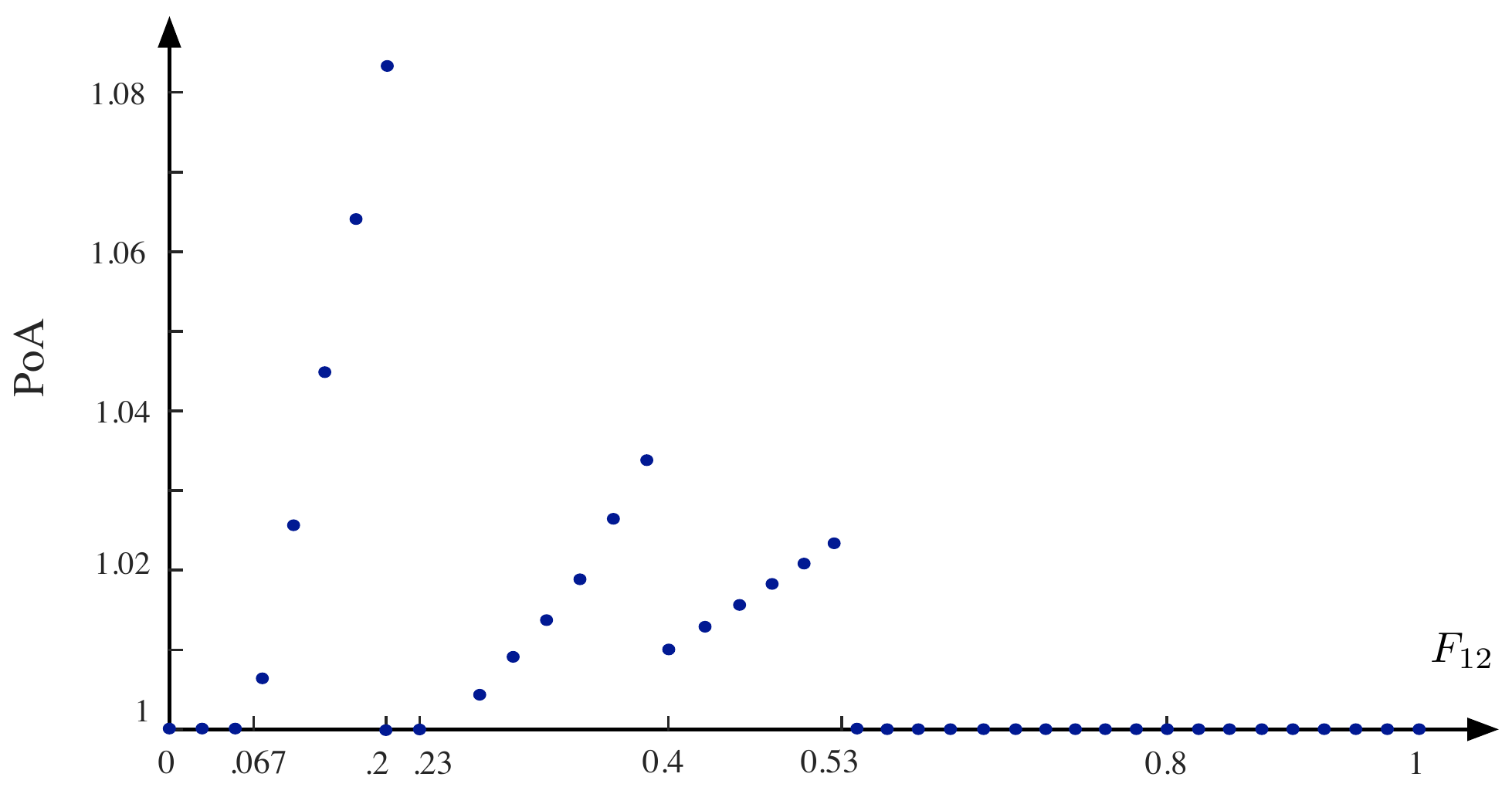}}
		\caption{\small Plot of price of anarchy as a function of $F_{12}$ for two groups of sizes 3 and 5, $\delta=0.5$, $c=0.2$.}\label{fig:PoA}
\end{figure}
We next discuss the evolution of intra vs. interconnections according to Formation Dynamics, when starting from an empty graph.

\begin{remark}[Formation of intra- vs. interconnections] Starting from an empty (or sparse) graph, initially the speed of formation of intra-connections is generally higher than that of interconnections. Fig.~\ref{fig:intra_inter} illustrates the fact that at the beginning of the formation dynamics, a certain amount of intra-connections is required to produce enough incentives for both groups to communicate through an interconnection.
\end{remark}
\begin{figure}[]
	\centering
	\includegraphics[width=0.9\linewidth]{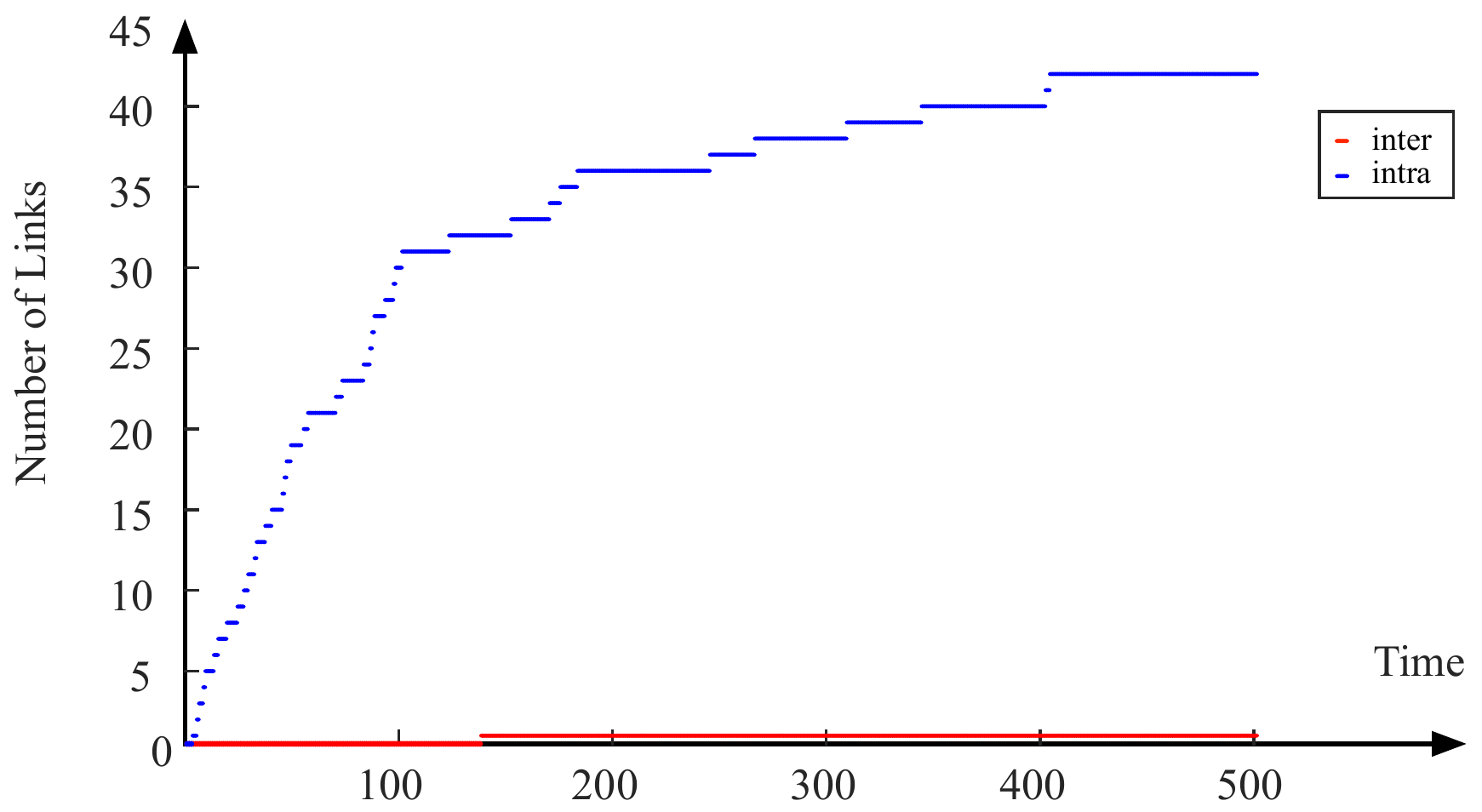}
        \caption{Time evolution of the number of inter- and intra-connections for two cliques of size 7.} \label{fig:intra_inter}
\end{figure} 

\section{Multigroup Connectivity Structure}

In what follows we analyze the more general case of having more than two groups. Consider $n$ individuals partitioned into groups $P_1,\dots,P_m$, and the payoff function defined above with 1-benefit $\delta<1$ and edge cost $c$, and functions $y_1$, $y_2$, and $y_3$. From Theorem \ref{thm: Pairwise Stability of Cliques} we know that each group forms a clique when $y_3>c$. We introduce the undirected graph $\mathcal T=(\mathcal V_{P}, \Edges_{\mathcal T} )$, whose nodes represent groups and $(\alpha,\beta)\in \Edges_{\mathcal T} $ if there exists at least one connection between $P_{\alpha}$ and $P_{\beta}$.

\begin{theorem}[Sufficiency Condition for Minimally Connected Cliques]\label{thm:multigroup_stability}
Consider $n$ individuals partitioned into groups $P_1,\dots,P_m$ of sizes $s_{1},\dots, s_{m}$ respectively. 
Assume that $c<y_3$ and $\mathcal T$ is connected. Then, there exists a pairwise stable structure consisting of minimally connected cliques along $\mathcal T$, if
 \begin{enumerate}
   \item \label{fact:nonedge-equal} for all $(\alpha, \beta)\in \Edges_{\mathcal T}$, $\alpha, \beta \in \until{m} $,  \\   
   \begin{equation}\label{utility-change-1}
   \begin{aligned}
  & \sum _{\lambda\neq \alpha, \lambda=1}^{m} F_{\alpha \lambda} 
   ( \delta^{d'_{\alpha\lambda}}-\delta^{d_{\alpha\lambda}})(1+(s_{\lambda}-1)\delta)>c,  \\
  & \sum _{\lambda\neq \beta, \lambda=1}^{m} F_{\beta \lambda} 
   ( \delta^{d'_{\beta\lambda}}-\delta^{d_{\beta\lambda}})(1+(s_{\lambda}-1)\delta)>c, \\
 & F_{\alpha \beta} <\max_{s\in \{s_{\alpha},s_{\beta}\}}\frac{c}{y_{2}(s)};  \text{ and}
   \end{aligned}
   \end{equation}
  \item \label{fact:tree-equal} for all $(\alpha, \beta) \notin \Edges_{\mathcal T}$, $\alpha, \beta \in \until{m} $,   \\
     \begin{equation}\label{utility-change-2}
     \begin{aligned}
      \sum _{\lambda\neq \alpha, \lambda=1}^{m} F_{\alpha \lambda} 
   ( \delta^{d'_{\alpha\lambda}}-\delta^{d_{\alpha\lambda}})(1+(s_{\lambda}-1)\delta)<c,\\
   ~ or  \sum _{\lambda\neq \beta, \lambda=1}^{m} F_{\beta \lambda} 
   ( \delta^{d'_{\beta\lambda}}-\delta^{d_{\beta\lambda}})(1+(s_{\lambda}-1)\delta)<c,
     \end{aligned}
   \end{equation}
\end{enumerate}    
where $d'_{\mu \lambda}=d_{\mu \lambda}(\Edges_{\mathcal T}\cup \{(\alpha, \beta)\})$ and $d_{\mu \lambda}=d_{\mu \lambda}(\Edges_{\mathcal T}\setminus \{(\alpha, \beta)\})$.
\end{theorem}
\begin{proof}
Suppose the network $\Edges_{0}$ (consisting of disjoint cliques) be connected along $\mathcal T$ and satisfy
\begin{enumerate}
\item if $(\alpha, \beta)\in \Edges_{\mathcal T}$, there is only one inter-link between $P_{\alpha}$ and $P_{\beta}$;
\item for every group $P_{\alpha}$, only one agent $i_{\alpha}$ has inter-links. 
\end{enumerate}
 For any pair of $(\alpha, \beta)$,  let $\Edges''=\Edges_{0}\cup\{(i_{\alpha},i_{\beta})\}$ and $\Edges'=\Edges_{0}\setminus\{(i_{\alpha},i_{\beta})\}$. it is easy to find that 
 \begin{align*}
U_i(\Edges'')-U_i(\Edges')=\sum_{k=1,k\neq \alpha}^{m} F_{k \alpha}\sum_{l\in P_{k}}(\delta^{d_{il}(\Edges'')}-\delta^{d_{il}(\Edges')}) -c.
\end{align*}
 Therefore, we have
 \[
 U_{i_{\alpha}}(\Edges'')- U_{i_{\alpha}}(\Edges')=
  \sum _{\lambda\neq \alpha, \lambda=1}^{m} F_{\alpha \lambda} 
   ( \delta^{d'_{\alpha\lambda}}-\delta^{d_{\alpha\lambda}})(1+(s_{\lambda}-1)\delta) -c
\]
and 
\[
  U_{i_{\beta}}(\Edges'')- U_{i_{\alpha}}(\Edges')=  \sum _{\lambda\neq \beta, \lambda=1}^{m} F_{\beta \lambda} 
   ( \delta^{d'_{\beta\lambda}}-\delta^{d_{\beta\lambda}})(1+(s_{\lambda}-1)\delta)-c.
 \]
Therefore, for $(\alpha, \beta)\in \Edges_{\mathcal T}$,  it follows from (\ref{utility-change-1}) that $\Edges''$ defeats $\Edges'$.  
 Similar to the proof of Theorem \ref{thm:existence-convergence-two}, since $F_{\alpha \beta} <\max_{s\in \{s_{\alpha},s_{\beta}\}}\frac{c}{y_{2}(s)}$, there only exists one inter-link between $P_{\alpha}$ and $P_{\beta}$.
 For $(\alpha, \beta)\notin \Edges_{\mathcal T}$, (\ref{utility-change-2}) implies that there exist no inter-link between $P_{\alpha}$ and $P_{\beta}$.
 
 By Theorem \ref{lem:stable-by-dynamics}, since network $\Edges_{0}$ can not be changed under dynamics, we can conclude that network $\Edges_{0}$ is stable.
\end{proof}

\begin{remark}
Theorem~\ref{thm:multigroup_stability} answers the question: given a certain matrix $F$ and graph structure $\Edges$ is $\Edges$ pairwise stable or not?
\end{remark}

\begin{corollary}
For the special case of interconnection structure being a star, with $P_{\gamma}$ as the central group, the sufficient condition of Theorem~\ref{thm:multigroup_stability} can be simplified as follows:\begin{enumerate}
   \item for all $\alpha \in \until{m} $, $(\alpha \neq \gamma)$ \\   
   \begin{equation*}
   \begin{aligned}
 & F_{\alpha \gamma} >\max_{s\in \{s_{\alpha},s_{\gamma}\}}\frac{c}{y_{1}(s)}; \text{ and}
   \end{aligned}
   \end{equation*}
  \item for all $(\alpha, \beta)  \in \until{m} $, $(\alpha, \beta \neq \gamma)$,   \\
  \begin{equation*}
   \begin{aligned}
 & F_{\alpha \beta} <\max_{s\in \{s_{\alpha},s_{\beta}\}}\frac{c}{y_{2}(s)}.
   \end{aligned}
   \end{equation*}
\end{enumerate}    
\end{corollary}

In the following example, we illustrate that due to randomness in choosing the pair of players, Formation Dynamic does not always converge to a unique stable structure even for the same initial network structure and matrix $F$.

\begin{example}\label{ex:multigroup-convergence}
Consider the case where we have $\dfrac{c}{y_{1}(s)}<F_{\alpha \beta}<\dfrac{c}{y_2(s)}$ for all $\alpha, \beta$, and that $\delta<\dfrac{\sqrt{5}-1}{2}$. We have five equal size groups named $\{1, 2, \dots, 5\}$. Two different processes are shown in Fig.~\ref{fig:dynamics_ex}. 
\begin{enumerate}[Process A.]
\item  For the process shown in Fig.~\ref{fig:dynamics_ex2} the order of pair selection is as follows: $(1, 2) \rightarrow (1, 3) \rightarrow (1, 4) \rightarrow (1, 4) \rightarrow (1,5)$. Note that by $(1,2)$ we mean an individual selected from group 1 paired with an individual from group 2, which results in the star graph being the convergent pairwise stable structure.
\item Now, consider the process shown in Fig.~\ref{fig:dynamics_ex1} for which the order of pairs of groups selected is as follows:  $(2, 3) \rightarrow (1, 2) \rightarrow (1, 3) \rightarrow (1,4)\rightarrow (3, 5) \rightarrow (2,4) \rightarrow (4,5)$. At the very last step, we have:
\[
\begin{aligned}
&U_i\big( \Edges \cup \{4,5\} -U_i(\Edges \setminus \{ 4,5\}) \big)\\
=& F_{45} \big(y_1(s) \big) - F_{45} \big( \delta^3 y_1(s) \big) \\
&\quad + F_{35} \big( \delta y_1(s) \big) - F_{35} \big( \delta^2 y_1(s)  \big)-c .
\end{aligned}\]
Since $F_{35}, F_{45}<\dfrac{c}{y_1(s)}$ and $\delta<\dfrac{\sqrt{5}-1}{2}$, we conclude that  
\[F_{45} (1-\delta^3) + F_{35} (\delta -\delta^2) > \dfrac{c}{y_1(s)}.
\]
 Consequently, we obtain 
 \[U_i\big( \Edges \cup \{4,5\} -U_i(\Edges \setminus \{ 4,5\}) \big)>0\]
  which means that the connection (4,5) is formed. Now since we have $F_{\alpha \beta}< \dfrac{c}{y_2(s)}$, no connected triad and thereby, no additional links will be formed. Also no link will be removed. Therefore, the final structure in Fig.~\ref{fig:dynamics_ex1}, which is a ring, is stable.

\end{enumerate}
\end{example}
 \begin{figure}[h]
	\begin{center} 
		\subfloat[Process A]{\includegraphics[width=0.9\linewidth]{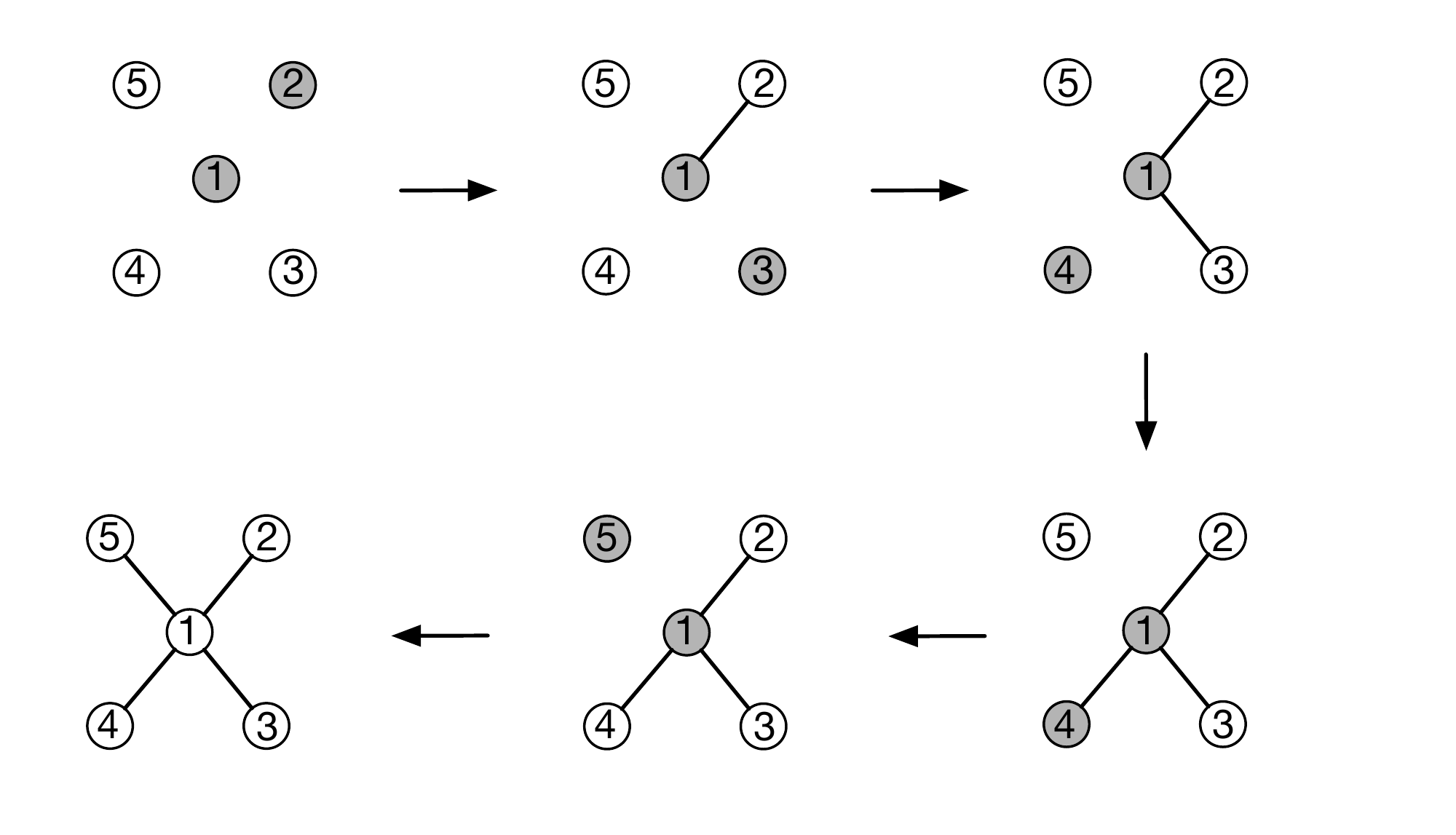}\label{fig:dynamics_ex2}}\quad
		\subfloat[Process B]{\includegraphics[width=0.9\linewidth]{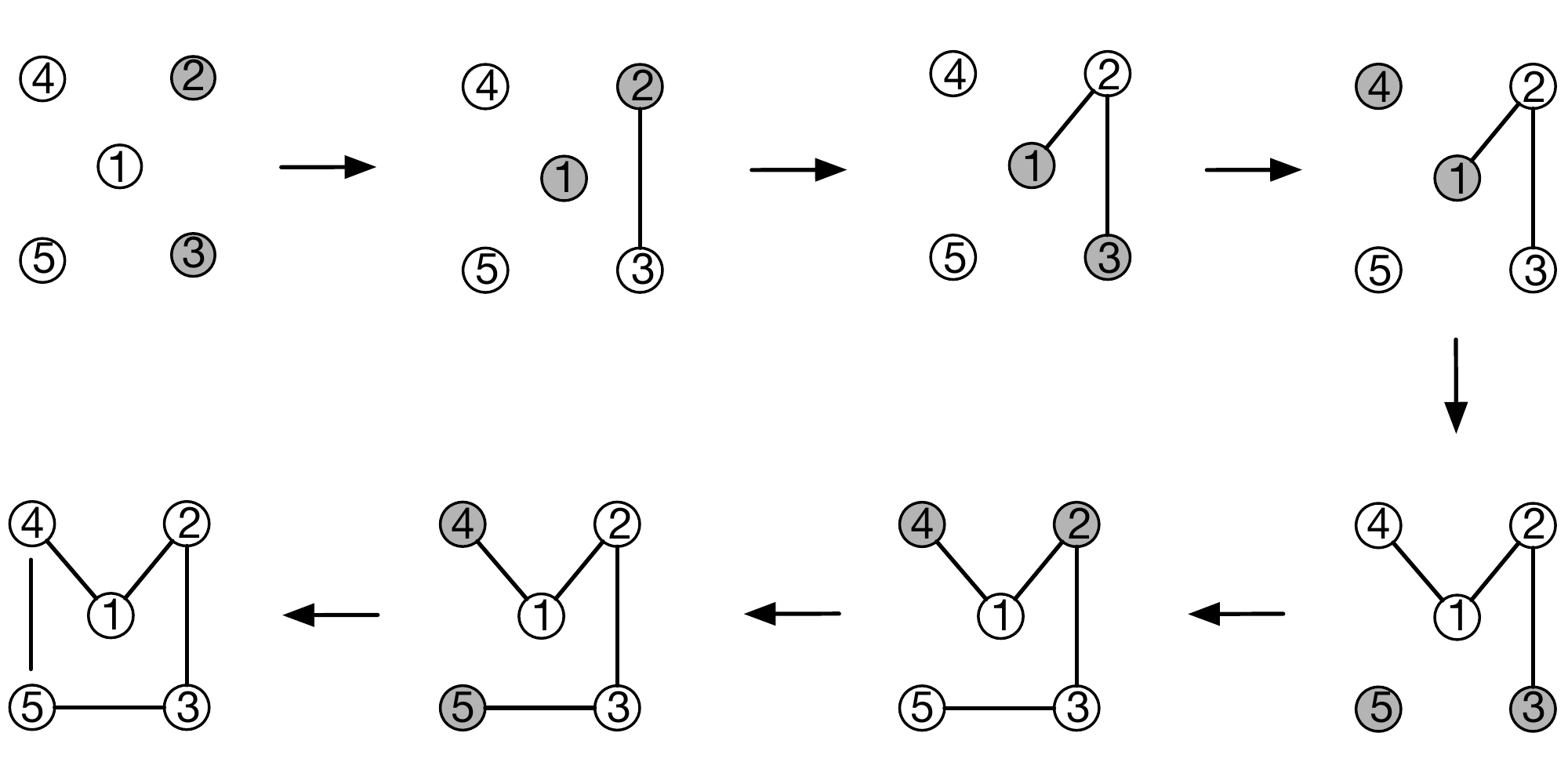}\label{fig:dynamics_ex1}}
		\caption{The processes of Example \ref{ex:multigroup-convergence}. At each step, shaded nodes represent the groups which the selected individuals belong to, and the outcome of the game (action taken regarding link addition, link removal, or indifference) is represented in the next.}\label{fig:dynamics_ex}
	\end{center}
\end{figure}

Example \ref{ex:multigroup-convergence} shows that, based on the order of the sequence of selected pairs, we can have two or possibly more convergent stable structures, and therefore, the convergence results cannot be generalized and the convergent structure is not always unique.

From Theorem \ref{thm: Pairwise Stability of Cliques} we know that each group forms a clique. We now analyze the interconnections among those cliques. Theorem~\ref{statements on redundant and comembers} addresses the redundancy of interconnections.

\begin{theorem}[Formation of Redundancies]\label{statements on redundant and comembers}
  Consider $n$ individuals partitioned into groups $P_1,\dots,P_m$ of sizes $s_1, s_2, \dots s_m$.  Suppose that  $c<y_3$. Then, under Formation Dynamics, 
  \begin{enumerate}
   \item\label{fact:redundant} redundant interconnections between $P_{\alpha}$ and $P_{\beta}$ will be formed and never removed,  if 
  $  F_{\alpha \beta} > \max\limits_{s \in\{s_{\alpha}, s_{\beta}\}} \dfrac{c}{ y_2(s)}$, and 
   \item\label{fact:comember} maximal interconnections between $P_{\alpha}$ and $P_{\beta}$ will be formed and never removed, if $\dfrac{c}{ y_3 }<F_{\alpha \beta} \leq 1$.      
  \end{enumerate}  
\end{theorem}

Fig.~\ref{fig:bounds-samesize} illustrates the four scenarios for $F_{\alpha \beta}$'s. The horizontal axis corresponds to the values of $F_{\alpha \beta}$ where $\{\alpha \beta\}$ belongs to the edge-set of the spanning tree, and the vertical axis corresponds to all other $F_{\alpha \beta}$'s. 

\begin{figure}[ht]
	\centering
	\includegraphics[width=0.99\linewidth]{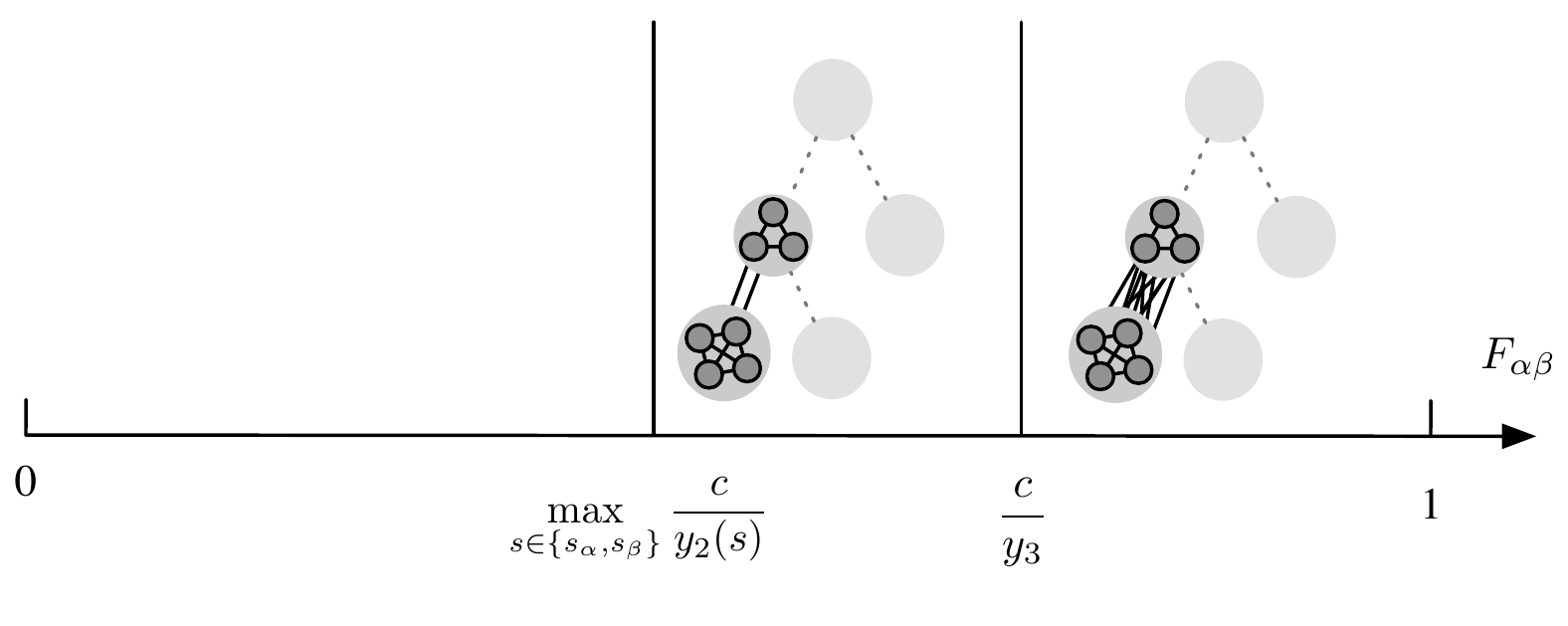}
        \caption{An illustration of ranges of parameter space in Theorem~\ref{statements on redundant and comembers} 
        } \label{fig:bounds-samesize}  
       \end{figure}

\begin{figure}[ht]
	\begin{center} 
		\subfloat[Number of Interconnections]{\includegraphics[width=0.8\linewidth]{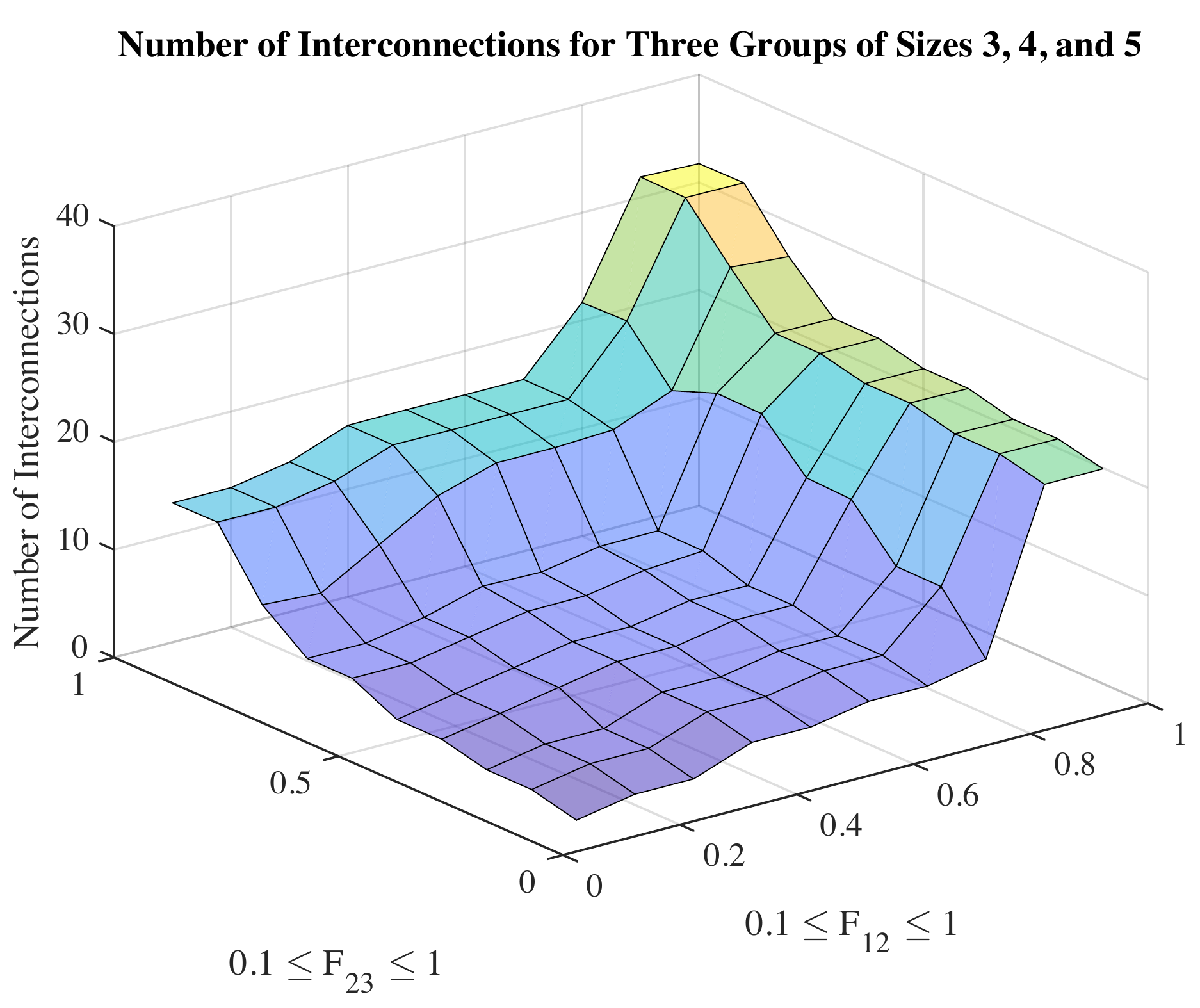}\label{}}\\
		\subfloat[Social Welfare]{\includegraphics[width=0.8\linewidth]{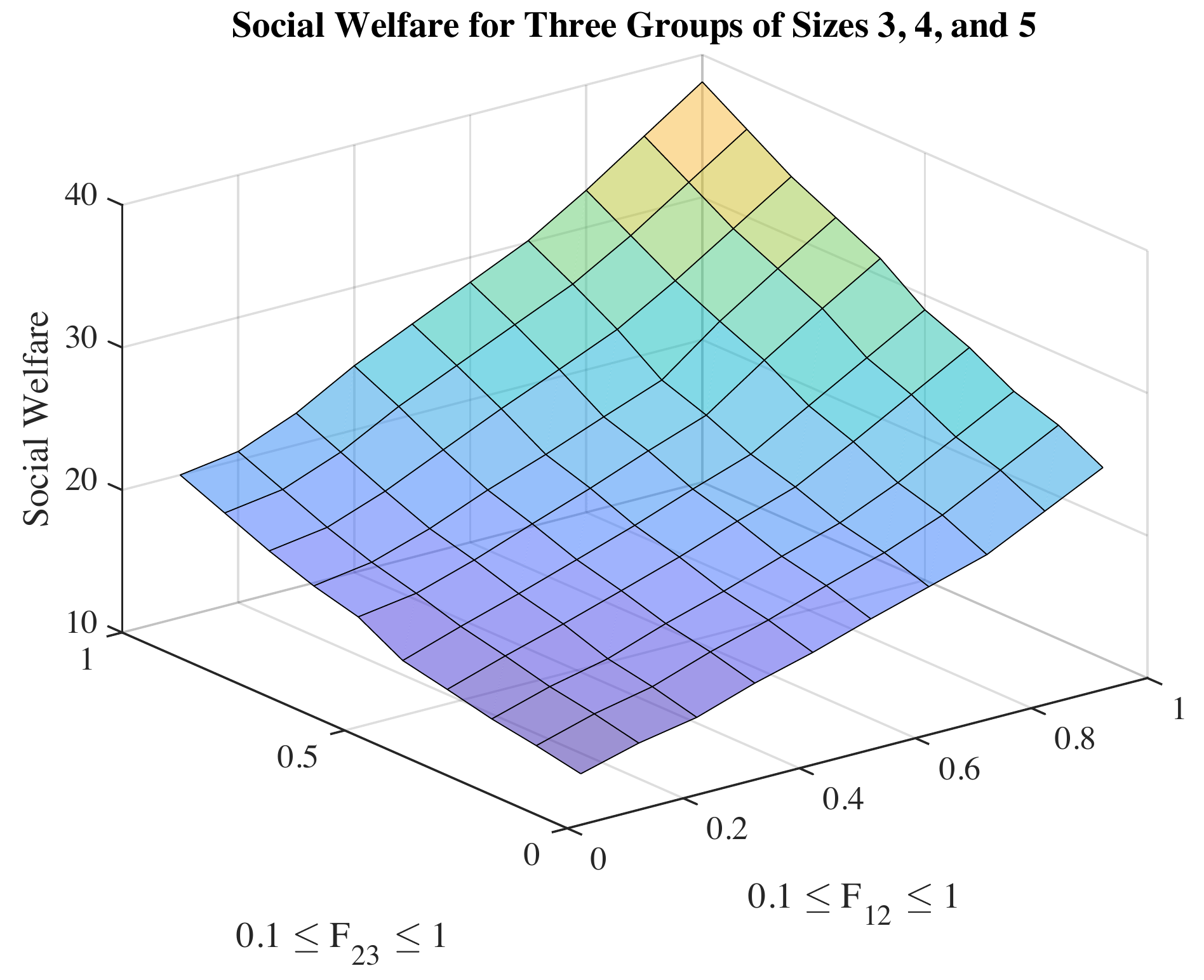}\label{}}
		\caption{\small Three groups of sizes 3, 4, and 5, $\delta=0.5$, $c=0.2$}\label{fig:3x3_general}
	\end{center}
\end{figure}

\begin{theorem}\label{thm:efficiency1} (Efficiency) For $n$ individuals partitioned into groups according to $P$, the efficient structure requires the same node to be chosen from each clique to provide bridges to other cliques, i.e., for a fixed interconnection structure and number of nodes, choosing the same representative from each clique increases the social welfare. 
\end{theorem}
\begin{proof}
Suppose that the density and the structure of interconnections are fixed. It is easy to see that by choosing the same representative from each clique, the distance between individuals from different cliques would be shortened, resulting in the term $\delta^d_{ij}$ in equations~\eqref{payoff-function} being larger. Therefore, the social welfare will increase. 
\end{proof}

However, this structure is unlikely to be pairwise stable since the
representative would bear a high cost for maintaining these
interconnections. Figure~\ref{fig:eff-thm} illustrates two networks with
the same interconnection structureswhere one network has higher social
welfare due to each group having only one representative.

\begin{figure}[ht]
	\begin{center} 
		\includegraphics[width=0.95\linewidth]{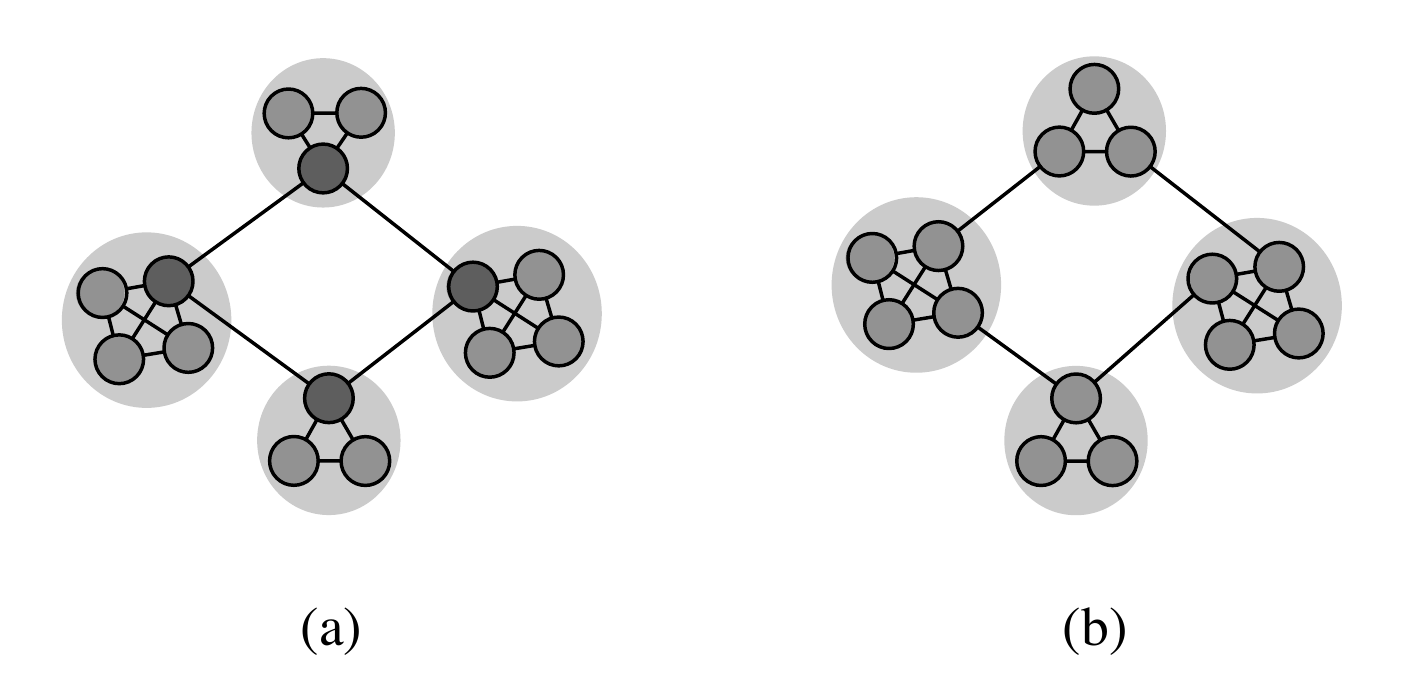}\label{fig:efficient-line}		\caption{\small In network a) each cliques has only one representative, whereas in figure b) some cliques have multiple representatives.}\label{fig:eff-thm}
	\end{center}
\end{figure}

\section{Conclusion}   
This paper proposes the first strategic network formation model that, given a matrix specifying the frequency of a coordination problem among groups, identifies the conditions that result in multigroup formation. The model deviates from the seminal papers on strategic network formation in that it accounts for heterogeneous frequency of control problems arising among the individuals and investigates pairwise stability and efficiency of multigroup connectivity structures, as well as convergence of the formation dynamics. 

In our model link formations occur bilaterally and thus many of the classical game-theoretic concepts do not apply to our framework. In particular, to study equilibrium structures, we utilize the concept of pairwise stability. A key challenge in our problem stems from the fact that not many tools are available for rigorous analysis or that they cannot be applied to the case of heterogeneous coordination problems among groups.

We identified the ranges of parameters where pairwise stable and efficient structures do and do not coincide and concluded that, for two-group structures, the efficient structures always has the same or a larger number of links than the pairwise stable ones. We also considered the price of anarchy and observed that the highest value occurs for the case when pairwise stable structures consist of disjoint union of cliques and the efficient structure has one link. Similar to the classical models, at the two ends of the spectrum of link values there is an overlap between efficient and stable structures.

We presented the conditions that result in the formation dynamics starting from an invariant set converge to cliques, and provided rigorous results for the number of interconnections in two-group structures. However, exact identification of the boundaries that result in certain number of interconnections among arbitrary number of groups with arbitrary size and interconnection structure is out of scope of this paper.

We note that by providing a full characterization of pairwise stability and efficiency for a two-group model, we focus on local topologies versus global topologies, as the individual interconnections can capture valuable information about the whole network and that all interconnections have subsets of two groups. This can be  interpreted into taking the distance only for the people in one's group or in the next immediate group in the utility function.

\section*{Acknowledgment}
The authors thank professors Ambuj K.\ Singh and Noah E.\ Friedkin for
their valuable comments and suggestions. This material is based upon work
supported by, or in part by, the U.S.~Army Research Laboratory and the
U.S.~Army Research Office under grant numbers W911NF-15-1-0577, and the National Natural Science Foundation of China under grant 61963032.

\ifCLASSOPTIONcaptionsoff
  \newpage
\fi

\bibliographystyle{plainurl+isbn}
\bibliography{alias,Main,FB,New}

\begin{thebibliography}{10}

\bibitem{VB-SG:00}
V.~Bala and S.~Goyal.
\newblock A noncooperative model of network formation.
\newblock {\em Econometrica}, 68(5):1181--1229, 2000.
\newblock \href {http://dx.doi.org/10.1111/1468-0262.00155}
  {\path{doi:10.1111/1468-0262.00155}}.

\bibitem{SPB-DSH:11}
S.~P. Borgatti and D.~S. Halgin.
\newblock Analyzing affiliation networks.
\newblock {\em The {Sage} Handbook of Social Network Analysis}, pages 417--433,
  2011.
\newblock \href {http://dx.doi.org/10.4135/9781446294413.n28}
  {\path{doi:10.4135/9781446294413.n28}}.

\bibitem{BB-MB-FB-AG:10}
B.~Bringmann, M.~Berlingerio, F.~Bonchi, and A.~Gionis.
\newblock Learning and predicting the evolution of social networks.
\newblock {\em IEEE Intelligent Systems}, 25(4):26--35, 2010.
\newblock \href {http://dx.doi.org/10.1109/MIS.2010.91}
  {\path{doi:10.1109/MIS.2010.91}}.

\bibitem{GCC-JSS:08}
G.~C. Chasparis and J.~S. Shamma.
\newblock Efficient network formation by distributed reinforcement.
\newblock In {\em {IEEE} Conf.\ on Decision and Control}, pages 1690--1695,
  Cancun, Mexico, 2008.
\newblock \href {http://dx.doi.org/10.1109/CDC.2008.4739163}
  {\path{doi:10.1109/CDC.2008.4739163}}.

\bibitem{GCC-JSS:13}
G.~C. Chasparis and J.~S. Shamma.
\newblock Network formation: {Neighborhood} structures, establishment costs,
  and distributed learning.
\newblock {\em IEEE Transactions on Cybernetics}, 43(6):1950--1962, 2013.
\newblock \href {http://dx.doi.org/10.1109/TSMCB.2012.2236553}
  {\path{doi:10.1109/TSMCB.2012.2236553}}.

\bibitem{BC-JAH:04}
B.~Cornwell and J.~A. Harrison.
\newblock Union members and voluntary associations: {M}embership overlap as a
  case of organizational embeddedness.
\newblock {\em American Sociological Review}, 69(6):862--881, 2004.
\newblock \href {http://dx.doi.org/10.1177/000312240406900606}
  {\path{doi:10.1177/000312240406900606}}.

\bibitem{NEF:83}
N.~E. Friedkin.
\newblock Horizons of observability and limits of informal control in
  organizations.
\newblock {\em Social Forces}, 62(1):54--77, 1983.
\newblock \href {http://dx.doi.org/10.1093/sf/62.1.54}
  {\path{doi:10.1093/sf/62.1.54}}.

\bibitem{NEF:98}
N.~E. Friedkin.
\newblock {\em A Structural Theory of Social Influence}.
\newblock Cambridge University Press, 1998, ISBN 9780521454827.

\bibitem{MSG:73}
M.~S. Granovetter.
\newblock The strength of weak ties.
\newblock {\em American Journal of Sociology}, 78(6):1360--1380, 1973.
\newblock \href {http://dx.doi.org/10.1086/225469} {\path{doi:10.1086/225469}}.

\bibitem{MOJ-BWR:05}
M.~O. Jackson and B.~W. Rogers.
\newblock The economics of small worlds.
\newblock {\em Journal of the European Economic Association}, 3:617--627, 2005.
\newblock \href {http://dx.doi.org/10.1162/jeea.2005.3.2-3.617}
  {\path{doi:10.1162/jeea.2005.3.2-3.617}}.

\bibitem{MOJ-AW:02}
M.~O. Jackson and A.~Watts.
\newblock The evolution of social and economic networks.
\newblock {\em Journal of Economic Theory}, 106:265--295, 2002.
\newblock \href {http://dx.doi.org/10.1006/jeth.2001.2903}
  {\path{doi:10.1006/jeth.2001.2903}}.

\bibitem{MOJ-AW:96}
M.~O. Jackson and A.~Wolinsky.
\newblock A strategic model of social and economic networks.
\newblock {\em Journal of Economic Theory}, 71(1):44--74, 1996.
\newblock \href {http://dx.doi.org/10.1006/jeth.1996.0108}
  {\path{doi:10.1006/jeth.1996.0108}}.

\bibitem{YJ-YW-XJ-ZZ-XC:17}
Y.~Jia, Y.~Wang, X.~Jin, Z.~Zhao, and X.~Cheng.
\newblock Link inference in dynamic heterogeneous information network: A
  knapsack-based approach.
\newblock {\em IEEE Transactions on Computational Social Systems}, 4(3):80--92,
  2017.
\newblock \href {http://dx.doi.org/10.1109/TCSS.2017.2715069}
  {\path{doi:10.1109/TCSS.2017.2715069}}.

\bibitem{RL:67}
R.~Likert.
\newblock {\em The Human Organization: Its Management and Values}.
\newblock McGraw-Hill, 1967, ISBN 0070378517.

\bibitem{LM:19}
L.~Maccari.
\newblock Detecting and mitigating points of failure in community networks: A
  graph-based approach.
\newblock {\em IEEE Transactions on Computational Social Systems},
  6(1):103--116, 2019.
\newblock \href {http://dx.doi.org/10.1109/TCSS.2018.2890483}
  {\path{doi:10.1109/TCSS.2018.2890483}}.

\bibitem{MMB:06}
M.~McBride.
\newblock Imperfect monitoring in communication networks.
\newblock {\em Economic Theory}, 126:97--119, 2006.
\newblock \href {http://dx.doi.org/10.1016/j.jet.2004.10.003}
  {\path{doi:10.1016/j.jet.2004.10.003}}.

\bibitem{SM-PA-FB-NEF:17c}
S.~Mohagheghi, P.~Agharkar, F.~Bullo, and N.~E. Friedkin.
\newblock Multigroup connectivity structures and their implications.
\newblock {\em Network Science}, pages 1--17, 2019.
\newblock \href {http://dx.doi.org/10.1017/nws.2019.22}
  {\path{doi:10.1017/nws.2019.22}}.

\bibitem{NO-FV:13}
N.~Olaizola and F.~Valenciano.
\newblock Network formation under linking constraints.
\newblock {\em Physica A: Statistical Mechanics and its Applications},
  392:5194--5205, 2013.
\newblock \href {http://dx.doi.org/10.1016/j.physa.2013.06.013}
  {\path{doi:10.1016/j.physa.2013.06.013}}.

\bibitem{NP-FD:19}
N.~Pagan and F.~D{\"o}rfler.
\newblock Game theoretical inference of human behavior in social networks.
\newblock {\em Nature Communications}, 10(1):5507, 2019.
\newblock \href {http://dx.doi.org/10.1038/s41467-019-13148-8}
  {\path{doi:10.1038/s41467-019-13148-8}}.

\bibitem{YS-MVDS:15}
Y.~Song and M.~{van~der~Schaar}.
\newblock Dynamic network formation with incomplete information.
\newblock {\em Economic Theory}, 59:301--331, 2015.
\newblock \href {http://dx.doi.org/10.1007/s00199-015-0858-y}
  {\path{doi:10.1007/s00199-015-0858-y}}.

\bibitem{WS-TE:08}
W.~Stam and T.~Elfring.
\newblock Entrepreneurial orientation and new venture performance: The
  moderating role of intra-and extraindustry social capital.
\newblock {\em Academy of Management Journal}, 51(1):97--111, 2008.
\newblock \href {http://dx.doi.org/10.5465/AMJ.2008.30744031}
  {\path{doi:10.5465/AMJ.2008.30744031}}.

\bibitem{MT-DK:10}
M.~Tortoriello and D.~Krackhardt.
\newblock Activating cross-boundary knowledge: The role of simmelian ties in
  the generation of innovations.
\newblock {\em Academy of Management Journal}, 53(1):167--181, 2010.
\newblock \href {http://dx.doi.org/10.5465/amj.2010.48037420}
  {\path{doi:10.5465/amj.2010.48037420}}.

\bibitem{HCW-SAB-RLB:76}
H.~C. White, S.~A. Boorman, and R.~L. Breiger.
\newblock Social structure from multiple networks. {I.} {B}lockmodels of roles
  and positions.
\newblock {\em American Journal of Sociology}, 81(4):730--780, 1976.
\newblock \href {http://dx.doi.org/10.1086/226141} {\path{doi:10.1086/226141}}.

\bibitem{JY-JL:12}
J.~Yang and J.~Leskovec.
\newblock Community-affiliation graph model for overlapping network community
  detection.
\newblock In {\em IEEE International Conference on Data Mining}, pages
  1170--1175, Brussels, Belgium, 2012.
\newblock \href {http://dx.doi.org/10.1109/ICDM.2012.139}
  {\path{doi:10.1109/ICDM.2012.139}}.

\bibitem{XZ-CW-YS-LP-HZ:17}
X.~Zhang, C.~Wang, Y.~Su, L.~Pan, and H.~Zhang.
\newblock A fast overlapping community detection algorithm based on weak
  cliques for large-scale networks.
\newblock {\em IEEE Transactions on Computational Social Systems},
  4(4):218--230, 2017.
\newblock \href {http://dx.doi.org/10.1109/TCSS.2017.2749282}
  {\path{doi:10.1109/TCSS.2017.2749282}}.

\end{thebibliography}

\end{document}